\documentclass[letterpaper,11pt]{article}

	\usepackage{a4,geometry}
\usepackage{graphicx}
\usepackage{amsmath,amssymb,amsthm,mathtools}
\usepackage{paralist}
\usepackage{bm}
\usepackage{bbm}
\usepackage{xspace}
\usepackage{url}
\usepackage{fullpage, prettyref}
\usepackage{boxedminipage}
\usepackage{wrapfig}
\usepackage{ifthen}
\usepackage{color}
\usepackage{xcolor}
\usepackage{framed}
\usepackage{algorithm}
\usepackage[pagebackref,letterpaper=true,colorlinks=true,pdfpagemode=none,urlcolor=blue,linkcolor=blue,citecolor=violet,pdfstartview=FitH]{hyperref}
\usepackage[nameinlink]{cleveref}
\usepackage{fullpage}
\usepackage{esvect}
\usepackage{mdframed}
\usepackage{thmtools}
\usepackage{thm-restate}
\usepackage{algpseudocode}

\usepackage{tikz} 
\usetikzlibrary{positioning}
\usetikzlibrary{arrows, decorations.pathmorphing}
\usetikzlibrary{decorations.pathreplacing}
\usepackage{standalone}

\usepackage{geometry}
\usepackage[pagebackref,letterpaper=true,colorlinks=true,pdfpagemode=none,urlcolor=blue,linkcolor=blue,citecolor=violet,pdfstartview=FitH]{hyperref}
\newcommand{\Hom}[2]{\mathrm{Hom}_{#2}(#1)}
\newcommand{\Sub}[2]{\mathrm{Sub}_{#2}(#1)}

\newcommand{\licl}{LICL}
\newcommand{\LICL}{LICL}

\newcommand{\UR}{UR}
\newcommand{\Reachable}{Reach}
\newcommand{\Unique}{U}

\newcommand{\Area}{A}
\newcommand{\Shared}{C}
\newcommand{\ReducedV}{\cV_{{\HubSet}'}}

\newcommand{\Frat}[2]{\Sigma(#1,#2)}
\newcommand{\MinFrat}[2]{{\vec{#1}}^{(#2)}}
\newcommand{\ExpG}{F}
\newcommand{\LabeledH}{H^{L}}
\newcommand{\Spasm}{Spasm}
\newcommand{\down}{down}
\newcommand{\ext}{ext}
\newcommand{\Gt}{G_t}
\newcommand{\Gtt}{G_{t'}}
\newcommand{\degenOrientation}{\vec{G}^\degen}
\newcommand{\optimalOrientation}{\vec{G}^*}
\newcommand{\HubSet}{\cS}
\newcommand{\Agg}{AGG}

\newcommand{\Extension}{Extension}
\newcommand{\Homomorphisms}{Homomorphisms}
\newcommand{\ComputeMinFrat}{CompOptimalExtension}
\newcommand{\ComputeFrat}{ComputeExtensions}

\newcommand{\SparsityAuthors}{Ne\v{s}et\v{r}il and Ossona de Mendez}

\newcommand{\degen}{\kappa}
\newcommand{\grad}{\nabla}
\newcommand{\topgrad}{\tilde{\nabla}}
\newcommand{\tw}{t}
\newcommand{\dtw}{\tau}
\newcommand{\htw}{\tau}
\newcommand{\dagtree}{DAG-tree decomposition}
\newcommand{\dagtreewidth}{DAG-treewidth}
\newcommand{\hubtree}{hub-tree decomposition}
\newcommand{\hubtreewidth}{hub-treewidth}
\newcommand{\TRICONJ}{Triangle Detection Conjecture}
\newcommand{\minor}[2]{#1\; \triangledown\; #2}
\newcommand{\topminor}[2]{#1\; \tilde{\triangledown}\; #2}

\newcommand{\ignore}[1]{}

\newcommand{\cB}{\mathcal{B}}
\newcommand{\cC}{{\cal C}}

\newcommand{\cE}{{\cal E}}

\newcommand{\cG}{\mathcal{G}}
\newcommand{\cH}{{\cal H}}

\newcommand{\cP}{\mathcal{P}}

\newcommand{\cS}{\mathcal{S}}
\newcommand{\cT}{{\cal T}}

\newcommand{\cV}{{\cal V}}

\newcommand{\Z}{{\mathbb Z}}

\newcommand{\NN}{\mathbb{N}}

\declaretheorem[name=Theorem, numberwithin=section]{theorem}
\declaretheorem[name=Lemma, numberlike=theorem]{lemma}
\declaretheorem[name=Claim, numberlike=theorem]{claim}
\declaretheorem[name=Corollary, numberlike=theorem]{corollary}

\declaretheorem[name=Fact, numberlike=theorem]{fact}
\declaretheorem[name=Proposition, numberlike=theorem]{proposition}

\declaretheorem[name=Definition, numberlike=theorem]{definition}

\newcommand{\Sec}[1]{\S \ref{sec:#1}} %section
\newcommand{\Eqn}[1]{\hyperref[eq:#1]{(\ref*{eq:#1})}} %equation
\newcommand{\Fig}[1]{{Fig.\,\ref{fig:#1}}} %figure
\newcommand{\Tab}[1]{\hyperref[tab:#1]{Tab.\,\ref*{tab:#1}}} %table
\newcommand{\Thm}[1]{\hyperref[thm:#1]{Theorem\,\ref*{thm:#1}}} %theorem
\newcommand{\Fact}[1]{\hyperref[fact:#1]{Fact\,\ref*{fact:#1}}} %fact
\newcommand{\Lem}[1]{\hyperref[lem:#1]{Lemma\,\ref*{lem:#1}}} %lemma
\newcommand{\Prop}[1]{\hyperref[prop:#1]{Prop.~\ref*{prop:#1}}} %property
\newcommand{\Cor}[1]{\hyperref[cor:#1]{Corollary~\ref*{cor:#1}}} %corollary
\newcommand{\Conj}[1]{\hyperref[conj:#1]{Conjecture~\ref*{conj:#1}}} %conjecture
\newcommand{\Def}[1]{\hyperref[def:#1]{Definition~\ref*{def:#1}}} %definition
\newcommand{\Alg}[1]{\hyperref[alg:#1]{Alg.~\ref*{alg:#1}}} %algorithm
\newcommand{\Clm}[1]{\hyperref[clm:#1]{Claim~\ref*{clm:#1}}} %claim
\newcommand{\Obs}[1]{\hyperref[obs:#1]{Observation~\ref*{obs:#1}}} %observation
\newcommand{\Rem}[1]{\hyperref[rem:#1]{Remark~\ref*{rem:#1}}} %remark
\newcommand{\Con}[1]{\hyperref[con:#1]{Construction~\ref*{con:#1}}} %construction
\newcommand{\Step}[1]{\hyperref[step:#1]{Step~\ref*{step:#1}}} %step
\newcommand{\Assumption}[1]{\hyperref[assm:#1]{Assumption\,\ref*{assm:#1}}} %assumption

\newcommand{\nesetril}{Ne\v{s}et\v{r}il}

\title{A Dichotomy Hierarchy Characterizing
Linear Time Subgraph Counting in Bounded Degeneracy Graphs}
\author{Daniel Paul-Pena\\
University of California, Santa Cruz\\
{\tt dpaulpen@ucsc.edu}
\and
C. Seshadhri\thanks{All the authors are supported by 
NSF CCF-1740850, DMS-2023495, and CCF-1839317.}\\
University of California, Santa Cruz\\
{\tt sesh@ucsc.edu}
}

\date{}

\begin{document}
\maketitle

\begin{abstract}
Subgraph and homomorphism counting are fundamental algorithmic problems. Given a constant-sized
pattern graph $H$ and a large input graph $G$, we wish to count the number of $H$-homomorphisms/subgraphs
in $G$. Given the massive sizes of real-world graphs and the practical importance of counting problems, we focus on
when (near) linear time algorithms are possible. The seminal work of Chiba-Nishizeki (SICOMP 1985)
shows that for bounded degeneracy graphs $G$, clique and $4$-cycle counting can be done linear time.
Recent works (Bera et al, SODA 2021, JACM 2022) show a dichotomy theorem characterizing the patterns $H$
for which $H$-homomorphism counting is possible in linear time, for bounded degeneracy inputs $G$.
At the other end, \nesetril{} and Ossona de Mendez used their deep theory of ``sparsity" to define
bounded expansion graphs (which contains all minor-closed families). They prove that, for \emph{all}
$H$, $H$-homomorphism counting can be done in linear time for bounded expansion inputs.
What lies between? For a specific $H$, can we characterize input classes where $H$-homomorphism counting
is possible in linear time?

We discover a hierarchy of dichotomy theorems that precisely answer the above questions.
We show the existence of an infinite sequence of graph classes $\cG_0 \supseteq \cG_1 \supseteq \ldots \supseteq \cG_\infty$
where $\cG_0$ is the class of bounded degeneracy graphs, and $\cG_\infty$ is the class of bounded expansion graphs.
Fix any constant sized pattern graph $H$. Let $\licl(H)$ denote the length of the longest induced cycle
in $H$. We prove the following. If $\licl(H) < 3(r+2)$, then $H$-homomorphisms can be counted in linear time
for inputs in $\cG_r$. If $\licl(H) \geq 3(r+2)$, then (assuming fine-grained complexity conjectures) $H$-homomorphism
counting on inputs from $\cG_r$ takes $\Omega(m^{1+\gamma})$ time. (Here, $m$ denotes the number of input edges,
and $\gamma$ is some explicit constant.) We prove similar dichotomy theorems
for subgraph counting.
\end{abstract}

\newpage

\section{Introduction} \label{sec:intro}

Counting the number of small patterns in a large input graph is a central algorithmic technique
and widely used in both theory and practice~\cite{lovasz1967operations,chiba1985arboricity,flum2004parameterized,dalmau2004complexity,lovasz2012large,AhNe+15,curticapean2017homomorphisms,PiSeVi17,SeTi19,roth2020counting}.
We express this problem as \emph{homomorphism} or \emph{subgraph} counting. 
The pattern is a simple, typically constant sized graph
$H = (V(H), E(H))$. The input simple graph is denoted by $G = (V(G), E(G))$. An $H$-homomorphism is a map $f:V(H) \to V(G)$ that preserves edges.
So, $\forall (u,v) \in E(H)$, $(f(u), f(v)) \in E(G)$. If $f$ is an injection (so distinct vertices of $H$ are mapped
to distinct vertices of $G$), this map is a subgraph.
We use $\Hom{G}{H}$ (resp. $\Sub{G}{H}$) to denote the count of the distinct $H$-homomorphisms (resp. $H$-subgraphs). 

Homomorphism and subgraph counting have applications in
logic, graph theory, partition functions in statistical physics, database theory, and network science~\cite{chandra1977optimal,brightwell1999graph,dyer2000complexity,borgs2006counting,PiSeVi17,dell2019counting,pashanasangi2020efficiently}. 
The topic of computing $\Hom{G}{H}$ is itself a subfield of graph algorithms~\cite{itai1978finding,alon1997finding,brightwell1999graph,dyer2000complexity,diaz2002counting,dalmau2004complexity,borgs2006counting,curticapean2017homomorphisms,bressan2019faster,roth2020counting}. 
The simplest non-trivial case is when $H$ is a triangle, which has itself led to numerous papers.

When $H$ is part of the input size, the problem is exactly counting subgraph isomorphisms, which is $\mathbb{NP}$-hard.
In many applications, the pattern is small and fixed.
Let $n = |V(G)|$ and $k = |V(H)|$. Even when $H$ is a $k$-clique,
the problem of computing $\Hom{G}{H}$ is $\#W[1]$-hard when parameterized by $k$~\cite{dalmau2004complexity}.
So we do not expect $n^{o(k)}$ algorithms in general. 
Nonetheless, the $n^k$ barrier can be beaten for specific $H$. 
The breakthrough result of Curticapean-Dell-Marx proved that if $H$ has treewidth at most 2, then $\Hom{G}{H}$ can be computed in $\text{poly}(k) \cdot n^\omega$ time, where $\omega$ is the matrix multiplication constant~\cite{curticapean2017homomorphisms}. 
Their result also showed that algorithms and lower bounds easily translate between homomorphism counting to subgraph counting.
In the following discussion, we only refer to $\Hom{G}{H}$. But all our questions and answers, with suitable modification, apply to $\Sub{G}{H}$ as well.

Homomorphism and subgraph counting have wide applications in network science, and there is a large study of
practical algorithms for this problem (refer to tutorial~\cite{SeTi19}). In practice, 
(near) linear time is likely a better mathematical abstraction for feasibility, than just polynomial time.
We are motivated by the following question.

\medskip

\emph{Under what conditions on $G$ and $H$ can $\Hom{G}{H}$ be computed in near-linear time?}

\medskip

A starting point for this broad investigation is a seminal result of Chiba-Nishizeki~\cite{ChNi85} that focuses
on \emph{graph degeneracy}. An input graph $G$ has bounded degeneracy, if all subgraphs of $G$ have bounded average degree.
Chiba-Nishizeki proved that clique counting and $4$-cycle counting can be done in linear time for bounded degeneracy graphs.
% Importantly, this result introduced the technique of
% graph orientations for subgraph counting, that has been at the center of many 
% state-of-the-art practical  algorithms~\cite{AhNe+15,jha2015path,PiSeVi17,ortmann2017efficient,jain2017fast,pashanasangi2020efficiently}. 
The degeneracy has a special significance in the analysis of real-world graphs, since it is intimately tied to the technique of ``core decompositions''~\cite{Se23}. 
The family of bounded degeneracy graphs is quite rich, and includes all minor-closed families, bounded treewidth classes, and preferential
attachment graphs. Most real-world graphs tend to have small degeneracy (\cite{goel2006bounded,jain2017fast,shin2018patterns,bera2019graph,bera2020degeneracy}, also Table 2 in~\cite{bera2019graph}),
underscoring  the practical importance of this class.

A series of recent subgraph counting advances provide strong dichotomy theorems characterizing the patterns $H$ for which $\Hom{G}{H}$
can be computed in linear time, when $G$ has bounded degeneracy~\cite{bressan2019faster,bera2020linear,Bera2021,Bera2022}. 
Assuming fine-grained complexity conjectures, linear time algorithms exist iff the longest induced cycle of $H$ is strictly less than $6$.
This is a surprisingly precise characterization, even though the final linear time algorithm is quite intricate.

At the ``other end", early work by Eppstein showed that, for all fixed $H$, determining the existence of an $H$-homomorphsim is linear-time computable
if $G$ is planar \cite{eppstein95planar}. These results were extended to bounded genus graphs \cite{eppstein2000diameter}. In a grand generalization of these
results, \nesetril{} and Ossona de Mendez established the concept of \emph{bounded expansion} graph classes \cite{NeMe-grad}.
These classes are defined using the theory of shallow minors. Bounded expansion classes are quite broad, and include all bounded degree
graphs, include bounded tree-width graps, and all minor-closed families. Bounded expansion graphs form a strict subset of bounded degeneracy graphs.
They proved that for all fixed $H$, if $G$ has bounded expansion, then one can count $H$-homomorphism/subgraphs in linear time
(refer to Table 18.1 and Section 18.6 of~\cite{NeMe-book} and \cite{NeMe-grad2}). 
% All these results deal with the decision problem. \nesetril{} and Ossona de Mendez showed that the counting problem can also be solved in linear time (Section $18.6$ of \cite{NeMe-book}).

To summarize the above discussion, we have two ends of a spectrum. Assume some fine-grained complexity
conjectures on triangle counting. Suppose $G$ has bounded degeneracy.
Then $\Hom{G}{H}$ is linear-time computable iff the longest induced cycle of $H$ is strictly less than $6$.
On other hand, if $G$ has bounded expansion, then for all $H$, $\Hom{G}{H}$ can be computed in linear time.

What lies in between? Is there some class of graphs between bounded degeneracy and bounded expansion graphs where,
say, $9$-cycle homomorphisms can be counted in linear time? In the context of linear time algorithms,
what determines the hardness of $H$-homomorphism counting?

\begin{figure}
	\centering
	\includegraphics[width=\textwidth*4/5]{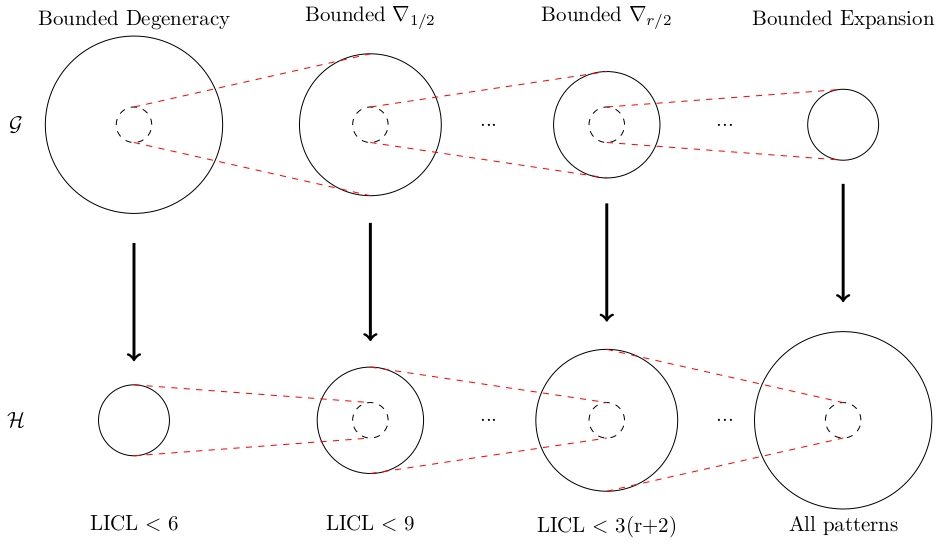}%
	\caption{A visualization of our main result. There is a decreasing hierarchy of input graph classes between bounded degeneracy and bounded expansion. 
    There is a corresponding increasing hierarchy of pattern classes, based on the $\licl$. As we look at more restrictive graph classes, we can count homomorphisms of
    more patterns in linear time.}
	\label{fig:main}
\end{figure}

\subsection{Main Result}

We give a comprehensive answer to the above questions. There is an infinite hierarchy of classes between bounded
degeneracy and bounded expansion graph classes. For any pattern $H$, we can precisely point out the largest
class where $\Hom{G}{H}$ is linear-time computable. 
We use $\licl(H)$ to denote the length of the longest induced cycle in $H$.

These graph classes are defined using a concept called the \emph{$r$ rank greatest reduced average degree} (or $r$-grad) of a graph~\cite{NeMe-book}.
The definition is technical and explained in the next section. For any $r \in \Z^+$, the quantity $\nabla_{r/2}(G)$ denotes
the $r/2$-grad of $G$. This is a well-defined graph quantity. Also, $\nabla_0(G)$ is the maximum average degree of any subgraph
of $G$, which, up to constant factors, is the graph degeneracy (or arboricity) (Theorem $4$ in \cite{Se23}). Moreover, for any $r < s$, $\nabla_{r/2}(G) \leq \nabla_{s/2}(G)$.

To give our main lower bound, we use the following common conjecture from fine-grained complexity, the \TRICONJ.

\begin{restatable}{conjecture}{triangleconjecture} \label{conjecture:triangle} 
	(\TRICONJ~\cite{abboud2014popular})
	There exists a constant $\gamma > 0$ such that in the word RAM model of $O(\log{n})$ bits, any algorithm to detect whether an input graph on $m$ edges has a triangle requires $\Omega(m^{1+\gamma})$ time in expectation.
\end{restatable}

% Now, for our main result, which provides a series of strict dichotomy theorems showing the exact cases where we can compute $\Hom{G}{H}$ in near-linear time.

\paragraph{The informal statement.} Consider a class of input graphs with bounded $\nabla_{r/2}$. If $\licl(H) < 3(r+2)$,
then $\Hom{G}{H}$ can be counted in linear time. If $\licl(H) \geq 3(r+2)$, then, assuming the \TRICONJ, any 
algorithm counting $\Hom{G}{H}$ for graphs with bounded $\nabla_{r/2}$ requires $\Omega(m^{1+\gamma})$ time. (Here, $m$ refers
to the number of edges in $G$.)

\begin{theorem} [Main Theorem] \label{thm:main}
    Fix any $r \in \Z^+$. Let $H$ be the pattern graph, and let the input graph $G$ have $m$ edges.
    Let $f: \NN \to \NN$ denote some explicit function. Let $\nabla_{r/2}(G)$ denote the $r/2$-grad of $G$.
	\begin{itemize}
        \item If $\LICL(H) < 3(r+2)$, then there exists an algorithm that computes $\Hom{G}{H}$ in time $f(\nabla_{r/2}(G)) \cdot m$.
		\item If $\LICL(H) \geq 3(r+2)$: Assume the \TRICONJ. For any function $g: \NN \to \NN$, there is no algorithm that computes $\Hom{G}{H}$
            in time $g(\nabla_{r/2}(G)) o(m^{1+\gamma})$ ($\gamma$ is the constant from the \TRICONJ).
	\end{itemize}
\end{theorem}

{\bf Remark:} The algorithm above is randomized, but the only use of randomness is in building hash tables for a polynomial
sized universe. Replacing the hash tables
with van Emde Boas trees, we can get a deterministic algorithm running in time $O(m\log\log m)$.

\medskip

\paragraph{The hierarchy for linear-time counting.} The theorem above can be informally visualized as \Fig{main}. Consider an infinite hierarchy
of nested graph classes $\cG_0 \supseteq \cG_1 \supseteq \cG_2 \ldots \supseteq \cG_\infty$, where
$\cG_\infty = \bigcap_{r \in \Z^+} \cG_r$. The class of bounded degeneracy graphs is $\cG_0$
and the class of bounded expansion graphs is $\cG_\infty$. (Recall that even $\cG_\infty$ contains all
minor-closed families; so it is really a big graph class by itself.)
Formally, $\cG_r$ is the class of graphs where $\nabla_{r/2}$ is bounded.

Now consider an ``opposite" hierarchy of pattern classes $\cH_0 \subseteq \cH_1 \subseteq \cH_2 \ldots \subseteq \cH_\infty$,
where $\cH_\infty$ is the set of all patterns. For every $r \in \Z^+ \cup \{\infty\}$, for all patterns $H \in \cH_r$,
there is a linear time algorithm computing $\Hom{G}{H}$ where $G \in \cG_r$. Moreover, for all $H \notin \cH_r$,
one requires $m^{1+\gamma}$ time to compute $\Hom{G}{H}$ for $G \in \cG_r$.

Specifically, for $r = 0$, $\cG_0$ is the class of bounded degeneracy graphs,
and $\cH_0$ is the set of patterns with $\licl(H) < 6$. 

\paragraph{The obstacle of long induced cycles.} \Thm{main} also gives a precise condition that makes patterns
harder to count. Long induced cycles are the obstruction towards efficient (near-linear)
algorithms. There is a curious jump of $3$ for the $\licl$ at every ``level" of this hierarchy.
While this may appear to be some artifact of the algorithm, this jump is matched by the hardness
results of \Thm{main}. We find it quite striking that the multiples of $3$ are exactly the transition points
for the hardness of homomorphism counting. The graph classes $\cG_r$ are defined by the $r/2$-grad
values, which seem to have no connection to these multiple of $3$ transition points.

\paragraph{The dichotomies for subgraph counting.} Subgraph counts can be easily represented as linear combinations
of homomorphism counts, using inclusion-exclusion. Hence, algorithms for the latter can be used for subgraph counting.
To count $H$-subgraphs, we count homomorphisms of all patterns formed by specific mergings of $H$.
Remarkably, a result of Curticapean-Dell-Marx showed that this procedure is actually optimal~\cite{curticapean2017homomorphisms}.
Meaning, \emph{lower bounds} for homomorphism counting translate to subgraph counting exactly as the 
upper bounds go. Using their techniques, we can adapt \Thm{main} to subgraph counting dichotomies.

For a pattern $H$, the $\Spasm(H)$ is the set of patterns obtained by merging any independent set of $H$.
(Note that an $H$-homomorphism may map an independent set to the same vertex of $G$.) Abusing
notation, let $\licl(\Spasm(H))$ denote the largest LICL value among all patterns in $\Spasm(H)$.
To get our hierarchical dichotomies for subgraph counting, we simply replace $\licl(H)$ in \Thm{main}
by the larger quantity $\licl(\Spasm(H))$.

\begin{theorem} [Dichotomies for subgraph counting] \label{thm:main-sub}
    Fix any $r \in \Z^+$. Let $H$ be the pattern graph, and let the input graph $G$ have $m$ edges.
    Let $f: \NN \to \NN$ denote some explicit function.
	\begin{itemize}
        \item If $\LICL(\Spasm(H)) < 3(r+2)$, then there exists an algorithm that computes $\Sub{G}{H}$ in time $f(\nabla_{r/2}(G)) \cdot m$.
        \item If $\LICL(\Spasm(H)) \geq 3(r+2)$: Assume the \TRICONJ. For any function $g: \NN \to \NN$, there is no algorithm with runtime that computes $\Sub{G}{H}$
            in time $g(\nabla_{r/2}(G)) o(m^{1+\gamma})$ ($\gamma$ is the constant from the \TRICONJ).
	\end{itemize}
\end{theorem}

\subsection{Shallow Minors and Greatest Reduced Average Density} \label{subsec:grad}

To formally explain what $\nabla_{r/2}$ means, we introduce the fundamental concept of \emph{shallow minors}.
Recall that a minor of $G$ is a graph $F$ formed as follows. Each vertex of $F$ represents a connected
subgraph of $G$. All of these connected graphs are vertex disjoint. An edge in $F$ represents
an edge in $G$ connecting the corresponding subgraphs. (Usually, a minor is described in terms 
of deletions and contractions. The connected subgraphs described above are contracted to the single
vertices of $F$.)

In a shallow minor at depth $d$, the connected subgraphs have radius at most $d$.
This section is taken from Sections 4.2 and 4.4 of~\cite{NeMe-book}.

\begin{definition}  \label{def:shallow_minor}
    The graph $G'$ is a \emph{shallow minor of $G$ at depth $d$} if there exists a collection of disjoint subsets $V_1,\ldots,V_p$ of vertices in $G$ such that:
	\begin{itemize}
		\item Each graph induced on $V_i$ has radius at most $d$: in set $V_i$, there is a vertex $x_i$ such that every vertex in $V_i$ is at distance at most $d$ from $x_i$ in the graph induced in $V_i$.
            This $x_i$ is called the \emph{center} of $V_i$.
        \item $G'$ is a subgraph of the graph $G$ with $\cP$ contracted: each vertex $v$ of $G'$ corresponds to a set $V_{i(v)}$, and edge $(u,v)$ in $G'$ corresponds to two sets $V_{i(u)}$ and $V_{i(v)}$ linked by at least one edge.
	\end{itemize}

	We use $G'\; \in\; \minor{G}{d}$ to denote that $G'$ is a shallow minor of $G$ at depth $d$.
\end{definition}

We can also define shallow minor at half-integer depths. Suppose $G' \in \minor{G}{d}$. There is a subgraph
of $G$ that is a \emph{witness}, which essentially contains the subgraphs of radius $d$ induced by the $V_i$'s (corresponding to vertices of $G'$) 
connected by certain edges (corresponding to the edges of $G'$). The latter edges are called \emph{external} edges of the witness.
The graph induced by each $V_i$ is called a \emph{bush}.

\begin{definition} \label{def:shallow-half} A minor $G' \in \minor{G}{d}$ is said to have \emph{depth $d-(1/2)$} if the following holds.
    There exists a subgraph of $G$ witnessing the $G'$ minor, such that for every external edge $(i,j)$: 
    let $B_i$ and $B_j$ be the corresponding bushes containing $i$ and $j$ respectively. Let $x_i$ and $x_j$ be the corresponding centers.
    Then, either the distance of $x_i$ to $i$, or the distance of $x_j$ to $j$, is strictly less than $d$.
\end{definition}

Let us unpack this definition. Each bush is a graph of radius $d$. But to witness the minor $G'$, the external
edges need not be ``maximally far" from the centers. So the minor is considered to have depth $d-(1/2)$, less than just $d$.
While this may seem like an extremely technical condition, the half-integer depth minors play a crucial role in \Thm{main}.
To precisely capture the linear-time hardness of homomorphism counting, we need the classes defined through half-integer
depth minors.

We define the central concept of the greatest reduced average density.

\begin{definition} 
    Let $r$ be a non-negative half-integer. The rank $r$ \emph{greatest reduced average density (grad)} of a graph $G$ is defined as:
	\begin{align*}
		\grad_r(G) = \max_{G'\; \in\; \minor{G}{r}}{\left\{\frac{|E_{G'}|}{|V_{G'}|}\right\}}
	\end{align*}
\end{definition}

In words, the rank $r$ grad is the maximum average degree over all minors of $G$ of depth $r$.

The classes of the hierarchy defined by \Thm{main} (and \Fig{main}) are bounded $\nabla_{r/2}$ graph classes.
Consider the simple case of $r=0$. A depth $0$ minor is just a subgraph. The rank $0$ grad is
the maximum average degree over subgraphs, which is (up to constant factors) the graph degeneracy. 

A graph class with bounded $\nabla_0$ is a class where all subgraphs of graphs in the class have bounded average degree.
This is precisely $\cG_0$ in our hierarchy.
A graph class has \emph{bounded expansion} if $\nabla_r$ is bounded for all $r$. 

\section{Main Ideas} \label{sec:main}

Our result has many moving parts. In this section, we give a high-level overview with a focus on various
obstacles we faced. Many of the new concepts and definitions were introduced to overcome these obstacles. 
Our result is obtained from marrying techniques from three sources: first and foremost, the deep
theory of sparsity of \nesetril{} and Osana de Mendez~\cite{NeMe-book}, the \dagtreewidth\ of Bressan \cite{bressan2019faster,Bressan2021},
and the unique reachability and induced cycle obstructions of Bera et al. (denoted BPS and BGLSS) \cite{Bera2021,Bera2022}.
Our lower bounds are fairly direct adaptations of techniques from BPS and BGLSS.

\paragraph{Graph orientations.} Arguably, the starting point for any work on subgraph counting
related to graph degeneracy is the clique-counting work of Chiba-Nishizeki \cite{ChNi85}.
A number of results recognized that the Chiba-Nishizeki ideas can be recast in terms
of \emph{graph orientations}~\cite{matula1983smallest,Se23}. The primary challenge
for homomorphism counting on sparse graphs is the presence of high-degree vertices. Such vertices
kill any simple brute force BFS procedure to find homomorphisms. The idea
is to orient/direct the edges of $G$ into a DAG, such that \emph{outdegrees} are bounded. The hope is then
to search for homomorphisms/subgraphs in constant-radius outneighborhoods, which will have bounded size.

A natural approach is to find an (acyclic) orientation that minimizes the maximum outdegree.
The optimal quantity is called the \emph{graph degeneracy}, and
remarkably, there is a simple linear time procedure to find such a ``degeneracy orientation" \cite{matula1983smallest}.
Moreover, this simple algorithm is intimately connected with $\nabla_0(G)$; the degeneracy
is a $2$-approximation of $\nabla_0(G)$. In words, all subgraphs of $G$ have bounded average
degree iff the degeneracy orientation has bounded outdegree. And we have assumed that $\nabla_0(G)$
is bounded, so we can orient $G$ into a DAG $\vec{G}$ of bounded outdegree.

\paragraph{Homomorphism counting for bounded degeneracy graphs.} We now outline
the upper bound results of BPS and BGLSS, which fundamentally use Bressan's \dagtreewidth.
Let us refer to an $H$-homomorphism/subgraph as a \emph{match}.

Every $H$-match in $G$ forms some directed match in $\vec{G}$. We can enumerate over all
the (constant many) orientations $\vec{H}$ of $H$, and count $\vec{H}$-matches in $\vec{G}$.
So we reduce to a directed acyclic homomorphism counting problem. 

Suppose $\vec{H}$ has a single source vertex, so there is a rooted directed tree $\vec{T}$ spanning $\vec{H}$.
Since $\vec{G}$ has bounded outdegree, there are $O(n)$ $\vec{T}$-matches in $\vec{G}$. (Which can
be enumerated by a bounded depth outward BFS from each vertex.) We can enumerate over all these $\vec{T}$-matches,
and see which of them induce $\vec{H}$-matches. When $H$ is a clique, this recovers Chiba-Nishizeki's original algorithm.
Moreover, this is probably one of the best practical algorithms for small clique counting~\cite{PiSeVi17}.

The story gets interesting when $\vec{H}$ has multiple sources. In this case, $\vec{H}$ can be covered
by a collection of rooted trees, one from each source in $\vec{H}$. These rooted trees are ``fragments"
of $\vec{H}$, which can be pieced together to yield an $\vec{H}$-match. For each fragment $\vec{T}$,
we can enumerate all the $\vec{T}$-matches. The ``piecing together" requires a careful indexing of 
all these matches.

When two fragment trees $\vec{T}$ and $\vec{T'}$ share a vertex (in $\vec{H}$), we have to
select corresponding matches in $\vec{G}$ that share a vertex. It is challenging to index
the tree matches appropriately to retrieve the relevant matches that might lead to an $\vec{H}$-match.
A number of results designed ad hoc methods for orientations of various $H$ \cite{cohen2009graph,PiSeVi17, bera2020linear}. A breakthrough
was achieved by Bressan, who gave a systematic algorithm that indexes the fragments
to efficiently count $\vec{H}$-matches \cite{bressan2019faster}. He introduced a novel concept of \emph{\dagtree}.

For a given $\vec{H}$, the \dagtree{} is a tree $\cT$ where nodes represent bags of sources in $\vec{H}$.
Roughly speaking, each subtree of $\cT$ represents a subgraph of $\vec{H}$ formed by all vertices reachable from
the sources (in the bags) in $\cT$. The subgraphs represented by independent subtrees can be counted/indexed independently.
The non-trivial step is the ``merging" of matches of children subtrees in $\cT$. Suppose a node in $\cT$ has two
children, which represent the subgraphs $\vec{H_1}$ and $\vec{H_2}$. The parent node will represent a subgraph $\vec{H'}$
that contains $\vec{H_1}$ and $\vec{H_2}$. Roughly speaking, we construct $\vec{H'}$-matches by
extending $\vec{H_1}$ and $\vec{H_2}$-matches through some shared vertices. These shared vertices are reachable
from the sources in the bag represented by the parent node. The complexity of this step is determined
by the bag size. The \dagtreewidth{} $\tau$ is the size of the largest bag, and the running time is $O(n^\tau)$.
Relevant to us, when the \dagtreewidth{} is one, the algorithm runs in (near) linear time.

We stress that the process is highly non-trivial, and a lot of homomorphism information needs
to be ``compressed". A planar graph can have $\Theta(n^2)$ $4$-cycles, yet the above method
can count them exactly in linear time.

When is the \dagtreewidth{} of $\vec{H}$ one? This is precisely captured by BPS and BGLSS.
If $\licl(H) < 6$, then for all orientations of $\vec{H}$, the \dagtreewidth{} is one. 
The proof of this fact involves new concept of unique reachability graphs; but we defer
the discussion of this point later. 

The above summary gives the overall picture of proving the existence of linear time
algorithms for $H$-homomorphism counting on bounded degeneracy graphs, where $\licl(H) < 6$.

\begin{figure}
	\centering
	\includegraphics[width=\textwidth*1/2]{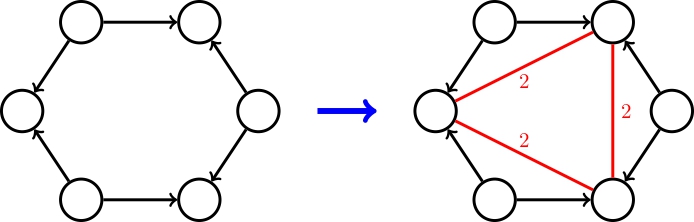}%
	\caption{On the left, an oriented $6$-cycle with $3$ sources. This orientation has a \dagtreewidth\ greater than $1$. It is not possible to decompose this pattern in a way that allows us to compute homomorphisms in linear time. The graph on the right is the result of connecting the endpoints of each out-out wedge, giving a fraternal augmentation. This new graph has a $\LICL < 6$ and hence for any orientation of the red edges, the \dagtreewidth\ is $1$.}
	\label{fig:6cycle}
\end{figure}

\subsection{The $6$-cycle obstruction} \label{sec:6-cycle}

We now explain the $6$-cycle barrier. Consider the oriented $6$-cycle $\vec{H}$ in the left of \Fig{6cycle}.
It can be partitioned into three out-out wedges (paths of length $2$), each corresponding
to a unique source. Thus, $\vec{H}$ forms a ``triangle" of out-out wedges. Counting $\vec{H}$-homomorphisms
is equivalent to counting triangles in the following graph. In the oriented $\vec{G}$, enumerate
all out-out wedges $(u,v,w)$, where $v$ denotes the wedge center. Create a new undirected graph ${G'}$ with the edges $(u,w)$. Since $G$ has bounded degeneracy, $\vec{G}$ has bounded
outdegree, and the number of out-out wedges is linear. So $G'$ has $O(m)$ edges.
Triangles in $G'$ are precisely $6$-cycles in $\vec{G}$.  Indeed, this argument gives the hardness
construction in BPS, reducing triangle counting in arbitrary graphs to $6$-cycle counting in bounded degeneracy graphs.

This is the starting point for our investigation. Under what circumstance can $6$-cycle counting be done in linear time?
If the graph $G'$ obtained above also had bounded degeneracy,
then triangle counting in $G'$ could be done in linear time (since $G'$ has $O(m)$ edges). 
What condition does $G$ need to satisfy for $G'$ to have bounded degeneracy?

\paragraph{Enter shallow minors.} Let us imagine contracting every alternate edge of the $6$-cycle. This leads to a triangle minor.
We can choose the centers of these contracted components with the following property. The
three non-contracted edges are incident to some center. Hence, this forms a shallow minor of depth $1/2$, according
to \Def{shallow-half}. Non-trivially, one can find a method of contracting $G$, so that all the $6$-cycles in $G$
are consistently contracted to triangles. Meaning, there is a $1/2$-shallow minor $G''$ such that $6$-cycles of $G$
become triangles in $G''$. The shallow minor machinery of \nesetril{} and Ossona de Mendez can be used
to show if $G''$ has bounded degeneracy, then the graph $G'$ (from the previous paragraph) also has bounded degeneracy.

Hence, if all $1/2$-shallow minors of $G$ have bounded degeneracy, then we can count triangles in $G'$ in linear time.
And the former condition is precisely saying that $\nabla_{1/2}(G)$ is bounded.

\paragraph{Implementing via fraternal augmentations.} Let us implement the above approach so that
it works for all $H$ with $\licl(H) = 6$. We start with $\vec{G}$ and $\vec{H}$ as before,
and assume that $\nabla_{1/2}(G)$ is bounded. We perform a series of \emph{fraternal augmentations}
in both $\vec{G}$ and $\vec{H}$. For every out-out wedge $(u,v,w)$, we add the edge $(u,w)$ to 
get the graphs $\vec{G'}$ and $\vec{H'}$. Note that new edges are undirected, so we try to orient
them in $\vec{G'}$ so that the maximum outdegree is minimized. Denote this graph as $\vec{G''}$. We then enumerate over all orientations
$\vec{H''}$ of the new edges in $\vec{H'}$. Finally, we count $\Hom{\vec{G''}}{\vec{H''}}$ and sum over all
the $\vec{H''}$.

Since $\nabla_{1/2}(G)$ is bounded, we can prove that $\vec{G''}$ will have bounded outdegree. We can also show that $\licl(\vec{H''})$,
treated as an undirected graph, will be strictly less than $6$. The key is that the augmentations in $\vec{H}$ will
reduce the length of \emph{all} induced cycles. This is seen for the simple example of the $6$-cycle in \Fig{6cycle}.
Hence, the previous machinery of BPS and Bressan using width one \dagtree s can be applied
to get a linear time algorithm. With some painstaking effort, one can push this approach to $\licl(H) < 8$.
Essentially, fraternal augmentations in $H$ reduce the $\licl$ to less than $6$, at which point previous methods
can run in linear time.

We note that the term ``fraternal augmentation" was introduced by \nesetril{} and Ossana de Mendes (Chap. 4 of~\cite{NeMe-book}).
But the idea is implicit in many previous results on subgraph counting in bounded degeneracy graphs~\cite{cohen2009graph,PiSeVi17,ortmann2017efficient}.

\begin{figure}
	\centering
	\includegraphics[width=\textwidth*1/3]{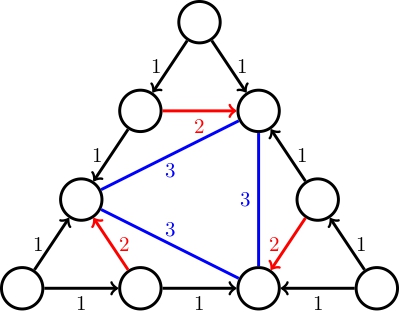}%
	\caption{An example of performing fraternal augmentations on an oriented $9$ cycle pattern (black edges). The first augmentation gives the red edges, reaching a situation analogous to the $6$ cycle in \Fig{6cycle}, an additional fraternal augmentation (blue edges) gives a pattern with an $\LICL$ less than $6$ and hence a \dagtreewidth\ of $1$.}
	\label{fig:9cycle}
\end{figure}

\subsection{More rounds of augmentations} \label{sec:more-aug}

Consider the oriented $9$-cycle pattern of \Fig{9cycle}. Let us perform a single round of fraternal augmentations,
to get the red edges. The $\LICL$ has now gone down to $6$, so we still cannot count homomorphisms (of the resulting pattern)
in linear time. Suppose we orient these new (red) edges, and perform another fraternal augmentation. This step
adds the blue edges, and the $\LICL$ is down to $3$. 

To count $9$-cycle homomorphisms in linear time, we need to perform two rounds of fraternal augmentations in $G$,
and hope that the degeneracy of the resulting graph is bounded. One might imagine that if $\nabla_1$ is bounded, then
two rounds of augmentations will lead to a bounded degeneracy graph. It turns out the situation is far more nuanced.
There are new obstacles for counting $9$-cycle homomorphisms in linear time.
This leads to the next technical tool.

\paragraph{Designing fraternity functions.} Augmentations are really shortcuts in the graph; each augmentation
represents a path of longer length. In general, we assume a bound on $\nabla_r(G)$ to get linear-time algorithms.
Such a bound refers to $r$-shallow minors, which essentially contract paths of length at most $2r$. Our augmentations
on such a graph should not shortcut a path that is longer than $2r$, since the bounded $\nabla_r(G)$ condition
cannot say anything about such augmentations. Thus, we have to perform augmentations carefully so that the bounded
$\nabla_r(G)$ condition can be used.

We discover the way to perform such careful augmentations is by crafting specific \emph{fraternity functions} of \nesetril{} and Ossona de Mendez.
This is a highly technical definition. At a high level, every augmented edge has a weight, which is (roughly) speaking
the length of the path shortcut by this edge. Any subsequent augmentation is not allowed to exceed a weight threshold.
We show an example of these weights in \Fig{9cycle}.
The final augmentation is described by a fraternity function, which satisfies a number of consistency constraints.
A deep result from the theory of sparsity is that if $\nabla_r(G)$ is bounded, then augmenting by a ``$(2r+1)$-fraternity function"
maintains bounded degeneracy. 
We apply this weighted fraternity function on both on the input $G$ and pattern $H$.

\begin{figure}
	\centering
	\includegraphics[width=\textwidth*1/2]{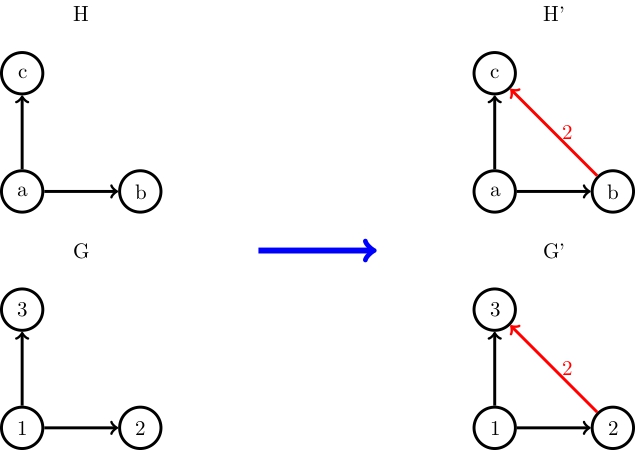}%
	\caption{An example of how fraternal augmentations do not preserve homomorphisms. Consider the homomorphism $\phi$ from $H$ to $G$ with $\phi(a)=1$ and  $\phi(b)=\phi(c)=2$. We can see how this will not be a valid homomorphism from $H'$ to $G'$ as the new edge connecting $b$ and $c$ is not preserved. However the number of subgraphs is preserved as the subgraphs $\{1,2,3\}$ in $G$ and in $G'$ are equivalent to $H$ and $H'$ respectively.}
	\label{fig:homomorphism}
\end{figure}

\paragraph{Maintaining homomorphism counts.} There are some annoyances when performing augmentations
for homomorphism counting. We explain these to motivate seemingly artificial technical conditions
in our homomorphisms and final counting algorithms.

As we add more edges to $\vec{G}$, we may create ``fake" $\vec{H}$ homomorphisms.
On the flip side, when augmenting $\vec{H}$, some existing matches may be inadmissible (due to new edges
in the pattern). We have a simple example in \Fig{homomorphism} where augmentations do not preserve
homomorphism counts.

We use two ideas to handle these problems. Firstly, we enforce that homomorphisms must be weight preserving,
where the weights come from the fraternity functions described earlier. This prevents mapping of augmented edges
to original edges and vice versa. Secondly, it is more convenient to create a new input instance
from a graph product $G \times H$. We find $H$-homomorphisms in this product graph, where it is much
easier to track the effect of augmentations on $H$-homomorphisms. All in all, we can then show direct
correspondences between homomorphisms in the original graph $G$, and the homomorphisms in the final graph
(constructed by graph products and a series of augmentation steps).

\begin{figure}
	\centering
	\includegraphics[width=\textwidth*1/2]{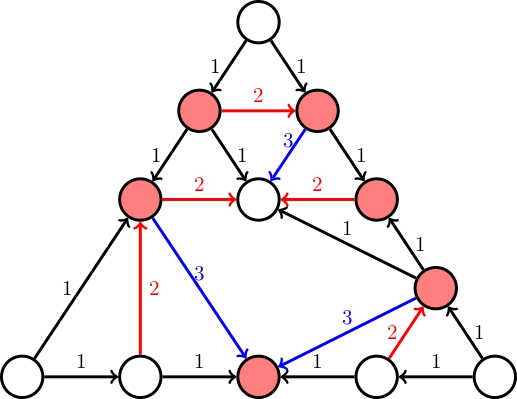}%
	\caption{An example of a pattern where we have performed 2 iterations of the augmentation (we will call this a $3$-fraternal extension of the pattern). As we can see the augmented pattern still has $\licl \geq 6$. However, the \hubtreewidth\ is $1$.}
	\label{fig:licl}
\end{figure}

\subsection{A major obstruction: extraneous induced cycles} \label{sec:extra}

So the overall story looks like the following. We have a graph $G$ such that $\nabla_r(G)$ is bounded.
We repeatedly perform augmentations, as long as they satisfy the constraints of a $(2r+1)$-fraternity function.
Intuitively, one can think of $2r$ rounds of augmentations. The resulting graph $G'$ has bounded degeneracy.
One also performs similar augmentations on $H$ to get the new pattern $H'$. (Of course, there is the extra complication of orienting every new edge
that is created, but let us ignore that for now.)

The hope is that $\licl(H')$ is strictly less than $6$, in which case previous algorithms can count $H'$-homomorphisms in $G'$
in linear time. Specifically, if $\licl(H) < 3(r+2)$, we would like $\licl(H')$ to be less than $6$. 
We think that with sufficiently many rounds of augmentations, we can cut down the $\licl$ length.

And this is false.
This statement fails, but only for a sufficiently complex example. The above approach does work for counting cycle
homomorphisms, or when $\licl(H) < 8$. But there is a pattern $H$ with $\licl(H) = 8$ where the approach breaks. 

For ease of exposition, we present an example with $\licl(H) = 9$.
Consider the pattern in \Fig{licl}. The problem is that the newly added augmentation edges (given in blue and red)
create a new induced cycle of length $6$. This induced cycle is given by the red vertices. Unfortunately,
we cannot guarantee a \dagtreewidth\ of one, so the existing algorithmic approach of BPS and BGLSS (as a black box) cannot yield
a linear time algorithm.

At this stage, the authors thought an entire rethink was needed. Thankfully, that was not needed.
The path around this obstruction is an unpacking of Bressan's algorithm and a deeper look
into the BPS machinery. By getting to the core of these results, we can generalize them appropriately
to deal with these ``extraneous" induced cycle in the patterns.

\paragraph{Dealing with cyclicity.} It turns out that a seemingly minor technicality
is important to handling \Fig{licl}. We started with a DAG $\vec{G}$ and 
a pattern $\vec{H}$. The reason to make $G$ into a DAG $\vec{G}$ was that the degeneracy orientation 
was linear time computable and gave a DAG with constant outdegree. As a result, the pattern $\vec{H}$ is also a DAG,
which motivated \dagtree{} and \dagtreewidth{}.

When we augment, we add new \emph{undirected} edges. To do a subsequent round of fraternal augmentations,
we need to orient these edges, so that we can construct new out-out wedges. Every orientation has to keep
the outdegree bounded. A natural approach is to extend the existing partial order (implied by the DAG $\vec{G}$).
This actually \emph{cannot} work. Meaning, if we want to keep the overall outdegree bounded after multiple augmentation
rounds, then we must use cyclic orientations of $G$. 

\paragraph{Hubsets to the rescue.} So we need to deal with cyclic patterns $\vec{H}$, while Bressan's algorithm is tailored to DAG patterns. 
Our insight is that Bressan's algorithm is quite flexible, and we can generalize the concept of DAG sources
to ``hubsets". A hubset is a set of vertices from which all other vertices can be reached. The corresponding
definitions of \dagtree{} and \dagtreewidth{} all generalize to hubsets. Technically, the proofs
of Bressan go through quite directly. But hubsets give us significantly more flexibility in minimizing
the ``hub treewidth".

Recall that the obstacle of \Fig{9cycle} has an induced cycle of length $6$, and is not guaranteed to have \dagtreewidth{} one. But we can argue than the hub-treewidth
is just one, which leads to a linear time algorithm for counting that pattern (when $\nabla_1(G)$ is bounded).

\paragraph{Extending BPS to hubsets.} 
In order to take advantage of the hubsets we have to rework the machinery of BPS 
that related induced cycles to \dagtreewidth{}, the process is quite technical. All in all, we can prove the following. If $\licl(H) < 3(r+2)$,
then (roughly speaking) after performing $r$ rounds of fraternal augments, the resulting pattern $\vec{H}$
has a hub treewidth of one. 

\subsection{Lower bounds and subgraphs} \label{sec:lower-main}

The lower bounds closely follow the techniques of BPS and BGLSS \cite{Bera2021,Bera2022}. Using 
the tensorization techniques of Curticapean, Dell, and Marx, one can essentially
show that the hardest patterns to count are cycles. The ideas in BPS and BGLSS 
are to use various graph products and manipulations, and they need to maintain
the degeneracy of their various constructions. In our setting, we deal with more
restrictive rank $r$ bounded grad graphs, so we need some extra care in our arguments.

The hardness for cycle counting is fairly straightforward, and taken from~\cite{bera2020linear}.
We basically subdivide an edge into a longer path, and reduce triangle counting in arbitrary graphs
to cycle counting in bounded grad graphs. We perform some calculations to show that the resulting
graphs has bounded grad. The rank $r$ determines the length of the subdivision, and hence
the length of the cycle that a triangle is converted to. It suffices to show that the final graph
has bounded rank $r$ grad, which is quite direct. These simple constructions match the upper bounds
of \Thm{main}, completing the story for homomorphism counting.

The deep insight of Curticapean, Dell, and Marx is that, as their title says, homomorphisms
form a good basis for subgraph counting \cite{curticapean2017homomorphisms}. Essentially, they show that for any quantity represented
as a linear combination of homomorphisms, the complexity of computing that quantity
is determined by the hardest homomorphism. It is fairly direct to see that $H$-subgraph counting
can be done by an inclusion-exclusion on the various $H'$-homomorphisms (for $H' \in \Spasm(H)$).
Using techniques from~\cite{curticapean2017homomorphisms}, we can translate the inclusion-exclusion
algorithm into hardness for subgraph counting.

\section{Related Work} \label{sec:related}

The theory of sparsity is a deep topic at the intersection of graph theory, logic, and combinatorics.
We refer the reader to the textbook~\cite{NeMe-book}. Chapters 4, 5, and 7 contain most of the relevant
background for our work. 

We cannot do justice to the literature on homomorphism counting, which has an immense history.
It was observed that the treewidth of the pattern plays a role in the final complexity.
D\'{\i}az et al.~\cite{diaz2002counting} designed an algorithm
with runtime $O(2^{k}n^{\tw(H)+1})$ where $\tw(H)$
is the treewidth of the target graph $H$.
Dalmau and Jonsson~\cite{dalmau2004complexity} proved that such a dependence on treewidth
is necessary. They show that that $\Hom{G}{H}$ is polynomial time solvable if and
only if $H$ has bounded treewidth, otherwise it is
$\#W[1]$-complete. 

Relevant to our framework of restrictions on both $G$ and $H$,
Roth and Wellnitz~\cite{roth2020counting} consider a doubly restricted version of $\Hom{G}{H}$, where both $H$ and $G$ are from 
graph classes.
They primarily focus on the parameterized dichotomy between
poly-time solvable instances and $\#W[1]$-completeness.
% Contrast to these works, we are interested in linear-time 
% dichotomies. Our focus is on characterizing the graph family $H$ 
% such that $\Hom{G}{H}$ admits a linear time algorithm when $G$ 
% belongs to the bounded degeneracy graph family.

% We give a brief review of the graph parameters treewidth and degeneracy.
% The notion of tree decomposition and treewidth were introduced in a seminal work by Robertson and Seymour~\cite{robertson1983graph,robertson1984graph,robertson1986graph}; although it has been discovered before under
% different names~\cite{bertele1973non,halin1976s}.
% Over the years, tree decompositions have been used extensively to 
% design fast divide-and-conquer algorithms for combinatorial problems.
 Degeneracy is a measure of sparsity and has
 been known since the early work of Szekeres-Wilf~\cite{szekeres1968inequality}.
 We refer to reader to a recent short survey of Seshadhri on subgraph counting and degeneracy~\cite{Se23}.
 The family of bounded degeneracy graphs is quite rich: it involves
 all minor-closed families, bounded expansion families, and preferential attachment
 graphs. Most real-world graphs have small degeneracy (\cite{goel2006bounded,jain2017fast,shin2018patterns,bera2019graph,bera2020degeneracy}, also Table 2 in~\cite{bera2019graph}).

Arguably the first work on exploiting degeneracy for subgraph counting is the seminal work of Chiba and Nishizeki~\cite{chiba1985arboricity}.
Since then, it has been a central technique in theoretical and practical algorithms~\cite{eppstein1994arboricity,AhNe+15,jha2015path,PiSeVi17,ortmann2017efficient,jain2017fast,pashanasangi2020efficiently}.

Bressan~\cite{bressan2019faster} introduced 
the concept of \dagtreewidth{} to design faster algorithms for 
homomorphism and subgraph counting problems in bounded degeneracy
graphs. Bressan showed that for a pattern $H$ with $|V(H)|=k$ and 
an input graph $G$ with $|E(G)|=m$ and degeneracy $\degen$,
one can count $\Hom{G}{H}$ in
$f(\degen,k)O(m^{\dtw(H)}\log m)$ time, where $\dtw(H)$ is the 
\dagtreewidth{} of $H$. Assuming the exponential time hypothesis~\cite{impagliazzo1998problems}, the subgraph counting problem does not admit
any $f(\degen,k)m^{o(\dtw(H)/\ln \dtw(H))})$ algorithm, for any positive function $f:\mathbb{N}\times \mathbb{N} \rightarrow \mathbb{N}$.
Recent work of Bressan, Lanziger, and Roth develops algorithms for pattern counting in directed graphs~\cite{BrLaRo23}.

A focus on linear time algorithm in bounded degeneracy graphs was initiated 
by Bera, Pashanasangi, and Seshadhri~\cite{bera2020linear}. They showed the lower
bound for counting $6$-cycles. That work was significantly generalized
by BPS and BGLSS which completely characterized linear time homomorphism counting
in bounded degeneracy graphs \cite{Bera2021,Bera2022}.

There are numerous pattern counting results
in Big Data models such as the property testing model~\cite{eden2017approximately,eden2018approximating,assadi2018simple,eden2020faster},
the streaming model~\cite{BarYossefKS02, Manjunath2011,Kane2012, ahn2012graph,Jha2013, Pavan2013, McGregor2016, bera2017towards,bera2020degeneracy}, and the Map Reduce model~\cite{cohen2009graph, Suri2011, kolda2014counting}.
% 
% problem of approximately counting homomorphism and 
% subgraphs have been studied extensively in various
% Big Data models such as the property testing model~\cite{eden2017approximately,eden2018approximating,assadi2018simple,eden2020faster},
% the streaming model~\cite{BarYossefKS02, Manjunath2011,Kane2012, ahn2012graph,Jha2013, Pavan2013, McGregor2016, bera2017towards,bera2020degeneracy}, and the map reduce model~\cite{cohen2009graph, Suri2011, kolda2014counting}.
% These works often employ clever sampling based techniques and forego exact algorithms.
% 

We now discuss the triangle detection conjecture.
Itai and Rodeh~\cite{itai1978finding}
gave the first non-trivial algorithm for the triangle 
detection and finding problem with $O(m^{3/2})$ runtime. 
The best known algorithm for the triangle detection problem uses fast matrix multiplication and runs in time
$O(\min \{n^\omega, m^{{2\omega}/{(\omega+1)}}\})$~\cite{alon1997finding}. If $\omega=2$, this yields a running time of $m^{4/3}$, which
many believe to be the true complexity. The current best is $O(m^{1.41\ldots})$, using the best matrix multiplication algorithms
Any improvement on this bound would be considered a huge breakthrough in algorithms.
Disproving the \TRICONJ{} would require an algorithm that would go even beyond the best possible matrix multiplication based algorithm.
We refer the reader to~\cite{abboud2014popular} for more details on \TRICONJ.

\section{Preliminaries} \label{sec:prelim}

\subsection{Graphs and Homomorphisms} \label{subsec:graphs}

We use $G = (V_G, E_G)$ to denote the input graph, we will use $n = |V_G|$ for the number of vertices of $G$ and $m = |E_G|$ for the number of edges. We use $H = (V_H, E_H)$ for the pattern graph and $k =|V_H|$ for the number of vertices of $H$, we consider $k$ to have constant value. Both graphs are simple and undirected.

We will also have labeled graphs, a labeled graph is a graph $G = (V_G,E_G,L_G)$, where $L_G: V_G \to S$ is the label function that maps the vertices of the graph to a set of labels $S$. Additionally we will have weighted labeled graphs $G = (V_G,E_G,W_G, L_G)$ where $W_G: E_G \to \mathbb{N}$ is the weight function that maps the edges of the graph to the correspondent weight. We will use $E_G^i$ to denote the the subset of edges of $G$ with weight equal to $i$, that is, $E_G^i= \{e \in E_{G} : W_{G}(e) = i\}$.

A homomorphism from $H$ to $G$ is a mapping $\phi: V_H \to V_G$ where $\forall(u,u') \in E_H$ we have $(\phi(u),\phi(u')) \in E_G$. We use $\Phi(H,G)$ for the set of homomorphisms from $H$ to $G$. We denote with $\Hom{G}{H}$ to the problem of counting the number of distinct homomorphisms from $H$ to $G$, that is $\Hom{G}{H} = |\Phi(H,G)|$. 

We extend these definitions for weighted and labeled graphs. Given two weighted labeled graphs $H',G'$ we will define a homomorphism from $H'$ to $G'$ as a mapping $\phi: V_{H'} \to V_{G'}$ such that $\forall\ u\in V_{H'}\ L_{H'}(u)= L_{G'}(\phi(u))$ and $\forall (u,v) \in E_{H'}$ we have $(\phi(u),\phi(u')) \in E_{G'}$ and $W_{H'}((u,v)) \geq W_{G'}((\phi(u),\phi(v)))$. Similarly $\Hom{G'}{H'}$ will correspond to the problem of counting the number of homomorphisms from $H'$ to $G'$.

We use $\LICL(H)$ for the largest induced cycle length of $H$, that is, the maximum length of any induced subgraph of $H$ that forms a cycle. We use $\Spasm(H)$ to refer to the spasm of $H$, that is, the collection of graphs obtained by contracting subsets of not-neighboring vertices in $H$. $\LICL(\Spasm(H))$ will be the largest induced graph in all the graphs in the spasm of $H$. 

\subsection{Subgraph copies}

Given the graphs $H$ and $G$ we say that $G'$ is a copy of $H$ in $G$ if $G'$ is a subgraph of $G$ such that there exists a 1:1 mapping from $H$ to $G'$ that preserves the edges. We use $\Sub{G}{H}$ for the problem of counting the number of distinct non-induced copies of $H$ in $G$.

There is a direct relation between $\Sub{G}{H}$ and $\Hom{G}{H'}$ for the graphs $H'\in \Spasm(H)$. The exact identity can be seen in \cite{curticapean2017homomorphisms}, but we can express it as follows:

\begin{lemma} \label{lem:sub}
	Given two graphs $G$ and $H$, for each graph $H_i \in \Spasm(H)$ there exists a non-zero constant $c_i$ such that:
	
	\begin{align*}
		\Sub{G}{H} = \sum_{H_i \in \Spasm(H)} c_i \Hom{G}{H_i}
	\end{align*}
\end{lemma}

%There is a direct relation between $\Sub{G}{H}$ and $\Hom{G}{H'}$ for the graphs $H'\in \Spasm(H)$, we use $H/\rho$ to define the quotient graph obtained from $H$ by merging a different sets of vertices in $\rho$, such that each set is forming an independent set into a single vertex. The collection of all possible $H/\rho$ graphs is exactly the $\Spasm(H)$. We have the following identity \cite{curticapean2017homomorphisms}:
%
%\begin{align} \label{eq:sub_to_hom}
%	\Sub{G}{H} = \sum_\rho\mu_\rho \cdot \Hom{G}{H/\rho}
%\end{align}
%
%With:
%\begin{align}
%	\mu_\rho = \frac{(-1)^{|V_H|-|V_{H/\rho}|}\cdot \prod_{B\in \rho}(|B|-1)!}{\Aut{H}}
%\end{align}
%Here $\Aut{H}$ is the number of automorphisms of the pattern graph$H$.

\subsection{Degeneracy and the degeneracy orientation}

A graph $G$ is $\degen$-degenerate if every subgraph of $G$ has a minimum degree of at most $\degen$. The degeneracy of $G$, $\degen(G)$, is the maximum value of $\degen$ such that $G$ is $\degen$-degenerate. 
The degeneracy is also called the \emph{coloring number}.
A graph has bounded degeneracy when $\degen$ has constant value. 

There is a way of orienting a graph acyclically, such that the maximum outdegree is upper bounded by the degeneracy of such graph.
This is a classic result in graph theory (refer to Section 5.2 of~\cite{Diestel-book} and survey~\cite{Se23}).
To construct such orientation, one can generate an ordering of the vertices of $G$ by iteratively selecting the lowest degree vertex in the graph and removing it. Hence the degeneracy orientation can be constructed in linear time.

\begin{fact} \label{fact:degen_orientation} (\cite{matula1983smallest})
	Given an undirected graph $G$ with degeneracy $\degen = \degen(G)$, there exists an acyclical orientation $\degenOrientation$ of $G$ such that the maximum outdegree of $\degenOrientation$ is $\degen$.
    Moreover, this orientation can be computed in time $O(n+m)$.
\end{fact}

For a directed graph $\vec{G}$, we will use $\Delta^+(\vec{G})$ to refer to the maximum outdegree of any vertex of $\vec{G}$.

% \begin{fact} \label{fact:degen_orientation} (\cite{matula1983smallest})
% 	Given a graph $G$ with degeneracy $\degen$, we can construct the degeneracy orientation $\degenOrientation$ of $G$ in time $O(n+m)$.
% \end{fact} 
% 
\subsection{Shallow Topological Minors and Top-Grads} \label{subsec:top-grad}

We can define shallow topological minors, which play a useful role in our analysis.

\begin{definition} [Shallow topological Minor \cite{NeMe-book}] \label{def:shallow_top_minor}
	A shallow topological minor of a graph $G$ of depth $d$ is a graph $G'$ obtained from $G$ by taking a subgraph and then replacing  an internally vertex disjoint family of paths of length at most $2d + 1$ by single edges. $G'$ is a shallow topological minor of $G$ at depth $d$ if there is a $\leq d $-subdivision of $G'$ that is a subgraph of $G$. 
	
	We use $G'\; \in\; \topminor{G}{d}$ to indicate that $G'$ is a shallow topological minor of $G$ at depth $d$.
\end{definition}

Similar to grads, we can define the the topological greatest reduced average density:

\begin{definition} (top-grad \cite{NeMe-book})
	The topological greatest reduced average density (top-grad) with rank $r$ of a graph $G$ is defined as:
	\begin{align*}
		\topgrad_r(G) = \max_{G'\; \in\; \topminor{G}{r}}{\left\{\frac{|E_{G'}|}{|V_{G'}|}\right\}}
	\end{align*}
\end{definition}

\SparsityAuthors\ proved that there is a polynomial relation between $\grad$ and $\topgrad$ of a graph $G$, this is given by the following fact:

\begin{fact} [Corollary $4.1$ in \cite{NeMe-book}]
	For every graph $G$ and every integer $r \geq 1$ holds
	\begin{align*}
		\topgrad_r(G) \leq \grad_r(G) \leq 4(4\topgrad_r(G))^{(r+1)^2}
	\end{align*}
\end{fact}

This corollary directly gives us the following fact.

\begin{fact} \label{fact:grad_topgrad}
	For any constant $r$, the classes of bounded rank $r$ grads and bounded rank $r$ top-grads are equivalent. Moreover, for any graph $G$, $\grad_r(G)$ is bounded if and only if $\topgrad_r(G)$ is bounded.
\end{fact}

\subsection{\dagtreewidth{} and Bressan's algorithm} \label{subsec:dag}

Bressan introduced the concepts of \dagtree\ and \dagtreewidth\ of a directed acyclic graph $\vec{H}$ \cite{bressan2019faster}. Before defining the \dagtree\ of $\vec{H}$ we need to define a few concepts. We use $S=S(\vec{H})$ to denote the set of sources of $\vec{H}$, that is, the vertices in $\vec{H}$ with no in-edges. Given two vertices $u,v \in V_{\vec{H}}$, we say that $v$ is reachable from $u$ if there exists a direct path in $\vec{H}$ from $u$ to $v$. 

For a vertex $s \in \vec{H}$ we use $\Reachable_{\vec{H}}(s)$ to denote the set of vertices of $\vec{H}$ that are reachable from $s$. We can extend this definition to set of vertices, let $B \subseteq S$ is a set of vertices of $\vec{H}$, we use $\Reachable_{\vec{H}}(B)=\bigcup_{s\in B} \Reachable_{\vec{H}}(s)$ for the union of the reachability sets of the vertices in the set $B$. Additionally we use $\vec{H}(s)$ to represent the induced subgraph of $\vec{H}$ in the vertices of $\Reachable_{\vec{H}}(s)$. 

We can now bring the definition of \dagtree\ of $\vec{H}$:

\begin{definition} [\dagtree \cite{bressan2019faster}]
	Let $\vec{H}$ be a directed acyclic graph with source set $S$. A \dagtree\ of $\vec{H}$ is a rooted tree $T=(\cB,\cE)$ with the following properties:
	\begin{enumerate}
		\item Each node $B\in \cB$ is a bag of sources, $B \subseteq S$.
		\item $\bigcup_{B \in \cB} B = S$.
		\item For all $B,B_1,B_2 \in \mathcal{B}$, if $B$ is on the unique path between $B_1$ and $B_2$ in $\mathcal{T}$, then we have $\Reachable_{\vec{H}}(B_1)\cap \Reachable_{\vec{H}}(B_2) \subseteq \Reachable_{\vec{H}}(B)$.
	\end{enumerate}
\end{definition}

Bressan also defined the \dagtreewidth:

\begin{definition} [\dagtreewidth ($\dtw$) \cite{bressan2019faster}]
	The \dagtreewidth\ of a \dagtree\ $\cT = (\cB, \cE)$ is defined as the maximum bag size over all the bags of sources of $\cT$:
	\begin{align*}
		\dtw(\cT) = \max_{B \in \cB} |B|
	\end{align*}
	We also use $\tau(\vec{H})$ to refer to the \dagtreewidth{} of the directed graph $\vec{H}$, which is the minimum $\tau(\cT)$ over all possible \dagtree\ $\cT$ of $\vec{H}$.
\end{definition}

Bressan introduced an algorithm that allows to compute the number of Homomorphisms between directed acyclical graphs $\vec{H}$ and $\vec{G}$. This algorithm uses the \dagtree\ $\cT$ of $\vec{H}$ to compute the homomorphisms from $\vec{H}$ to $\vec{G}$ from the homomorphisms of $\vec{H}(B)$ to $\vec{G}$ for each of the bags sources $B$ in $\cT$. When $\vec{G}$ has bounded outdegree and $\vec{H}$ has $\dtw(\vec{H})=1$ this will take $O(n \log n)$ time. We can restate this result as follows:

\begin{lemma} [\cite{bressan2019faster}]
	Let $\vec{H}$ be a directed acyclic graph with $\dtw(\vec{H})=1$ and let $\vec{G}$ be a directed acyclic graph with bounded maximum outdegree. There is an algorithm that computes $\Hom{\vec{G}}{\vec{H}}$ in $O(n\log{n}))$ time.
\end{lemma}

\subsection{Graph Products}  \label{subsec:products}

We will use two different graph products in this work, they are standard definitions but we show them here for completeness:

\begin{definition} (Categorical Product)[\cite{NeMe-book}] \label{def:categorical_product}
	Given two graphs $G$ and $G'$, we define their categorical product $G \times G'$ as follows:
	\[
	V_{G \times G'} = V_{G} \times V_{G'}
	\]
	\[
	E_{G \times G'} = \{(\langle u,u'\rangle, \langle v, v'\rangle): (u,v) \in E_{G} \text{ and } (u',v') \in E_{G'}\}
	\]
\end{definition}

\begin{definition} (Lexicographical Product)[\cite{NeMe-book}]
	Given two graphs $G$ and $G'$, we define their lexicographical product $G \bullet  G'$ as follows:
	\[
	V_{G \bullet  G'} = V_{G} \times V_{G'}
	\]
	\[
	E_{G \bullet G'} = \{(\langle u,u'\rangle, \langle v, v'\rangle): (u,v) \in E_{G} \text{ or } (u=v \text{ and }(u',v') \in E_{G'}\}
	\]
\end{definition}

We can show that there is a direct relation between both products:

\begin{fact} \label{fact:products}
	For any constant $c>0$: $G \times K_c = G \bullet \bar{K_c}$
\end{fact}
\begin{proof}
	First, from both definitions we have that the vertex sets are equivalent: $V_{G \times K_c} = V_{G \bullet  \bar{K_c}} = V_G \times V_{K_c}$.
	
	We can also show that the edge sets are equivalent:
	\begin{itemize}
		\item $E_{G \times K_c} = \{(\langle u,u'\rangle, \langle v,v'\rangle): (u,v) \in E_G \text{ and } (u',v') \in E_{K_c}\} = \{(\langle u,u'\rangle, \langle v,v'\rangle): (u,v) \in E_G\}$.
		
		\item $E_{G \bullet \bar{K_c}} = \{(\langle u, u'\rangle, \langle v, v'\rangle): (u,v) \in E_G \text{ or } (u=v \text{ and }(u',v') \in E_{\bar{K_c}}\} = \{(\langle u,u'\rangle, \langle v,v'\rangle): (u,v) \in E_G\}$.
	\end{itemize}
\end{proof}
\subsection{Fraternity Function} \label{subsec:fraternal}

 \SparsityAuthors\ introduced the notion of Fraternity Function in \cite{NeMe-book}:

\begin{definition} \label{def:frat_function}
	\textbf{(Fraternity Function) [Section 7.4 \cite{NeMe-book}]} Let $\mathcal{V}$ be a finite set and let $t$ be an integer. $t$-fraternity function is a function $\omega: \mathcal{V} \times \mathcal{V} \to \mathbb{N} \cup \{\infty\}$ such that for every $x,y \in \mathcal{V}$ one of $\omega(x,y)$ and $\omega(y,x)$ (at least) is $\infty$ and such that for every $x\neq y \in \mathcal{V}$:
	\begin{itemize}
		\item Either $\min(\omega(x,y),\omega(y,x))=1$
		\item Or $\min(\omega(x,y),\omega(y,x)) = \min_{z\in \mathcal{V}\setminus\{x,y\}} \omega(z,x) + \omega(z, y)$ \footnote{All the definitions in the Sparsity book use indegree instead of outdegree, which is more commonly used for subgraph counting using the degeneracy orientation. We will use the outdegree orientation instead and hence some of the definitions have been altered to reflect this.} 
		\item Or $\min(\omega(x,y),\omega(y,x)) > t$ and $\min_{z\in \mathcal{V}\setminus\{x,y\}} \omega(z,x) + \omega(z, y)> t$.
	\end{itemize}
\end{definition}

Given a $t$-fraternity function $\omega$ we define the directed weighted graph $\vec{G}^\omega = (V_{\vec{G}^\omega}, E_{\vec{G}^\omega}, W_{\vec{G}^\omega})$ as the graph with vertex set $V_{\vec{G}^\omega} = \mathcal{V}$ whose edges are all the pairs $u,v \in \mathcal{V}$ such that $\omega(u,v)\leq t$ and for every edge $(u,v)$ we have $W((u,v)) = \omega(u,v)$. We also define the directed graph $\vec{G}^\omega_i = (V_{\vec{G}^\omega_i}, E_{\vec{G}^\omega_i})$ as the graph with vertex set $V_{\vec{G}^\omega} = \mathcal{V}$ whose edges are all the pairs $u,v \in \mathcal{V}$ such that $\omega(u,v) = i$. We define $\Delta^+_i(\omega) = \Delta^+(\vec{G}^\omega_i)$.

We say that a directed weighted graph $\vec{G}$ forms a $t$-fraternity function if there is a $t$-fraternity function $\omega$ such that $\vec{G} = \vec{G}^\omega$. We call the fraternity function of $G$ to $\omega$ in such cases.

\section{Proving the Upper Bound} \label{sec:upper}

In this section we state the theorem that shows the upper bound of the Main Theorem and give a complete proof for it. In the following sections we will proof the different lemmas that compose the proof: 

\begin{theorem} \label{thm:upper}
	For all $t>0 \in \mathbb{N}$, let $G$ be an input graph with $n$ vertices, $m$ edges and bounded $\grad_{(t-1)/2}(G)$. If $\LICL(H) < 3(t+1)$ then exists an algorithm that computes $\Hom{G}{H}$ in $O(n \log{n})$ time.
\end{theorem}

\begin{proof}
	Fix any $t>0$. 
	
	First, we need to compute the labeled version of $H$ and $G$, respectively $\LabeledH$ and $\ExpG$. This is because the fraternal extension procedures that we will perform in the following step will not preserve the number of homomorphisms if applied directly on $H$ and $G$ (see \Fig{homomorphism}).
	
	The constructions of these graphs are defined in \Sec{label}. We can construct $\LabeledH$ in $O(k)$ time and $\ExpG$ in $O(n)$ time (\Clm{construct_G_1}). Additionally, we will show that this transformations preserves both the LICL of the pattern graph and the bounded grad conditions on the pattern and input graphs respectively:
	\begin{itemize}
		\item $\LICL(\LabeledH) = \LICL(H)$
		\item For any constant $i$, $\grad_i(G)$ is bounded if and only if $\grad_i(\ExpG)$ is bounded. (\Lem{labeled_grad}).
	\end{itemize}
	We will also prove that the number of homomorphisms from $H$ to $G$ is equal to the number of homomorphisms from $\LabeledH$ to $\ExpG$ (\Clm{hom_equivalence_1}).
	
	Then, we will compute both the optimal acyclic $t$-Fraternal extension of $\ExpG$, $\MinFrat{\ExpG}{t}$, and the collection of $t$-fraternal extensions of $\LabeledH$. We will formally define these concepts in \Sec{fraternal}. We will also show how to efficiently construct $\MinFrat{\ExpG}{t}$:
	
	\begin{restatable}{lemma}{buildingG} \label{lem:construct_G}
		Let $\ExpG$ be a graph with $O(n)$ vertices and bounded $\grad_{(t-1)/2}$, we can construct $\MinFrat{\ExpG}{t}$ in $O(t \cdot n)$ time.
	\end{restatable}

	Additionally, in the same section we prove that $\MinFrat{\ExpG}{t}$ will have bounded maximum outdegree:
	
	\begin{restatable}{lemma}{lemmagrad} \label{lem:main_grad}
		Let $\ExpG$ be a graph with $O(n)$ vertices and bounded $\grad_{(t-1)/2}$, then the directed graph $\MinFrat{\ExpG}{t}$ has bounded max outdegree.
	\end{restatable}

	$\Frat{H}{t}$ is independent on the input graph and will only depend on the pattern graph $H$. Because $H$ is assumed to be constant sized we will be able to compute $\Frat{H}{t}$ in $O(1)$ time. 
	
	We can now reduce our problem from computing $\Hom{\ExpG}{\LabeledH}$ to computing $\Hom{\MinFrat{\ExpG}{t}}{\vec{H}'}$ for all $\vec{H}'\in \Frat{H}{t}$. In Subsection \ref{subsec:equivalence} we will show a direct equivalence between both quantities, given by the following lemma:
	
	\begin{restatable}{lemma}{equivalence} \label{lem:equivalence}
		\[
		\Hom{\ExpG}{\LabeledH} = \sum_{\vec{H}'\in \Frat{H}{t}} \Hom{\MinFrat{\ExpG}{t}}{H'}
		\]
	\end{restatable}
	
	In \Sec{hub-set} we introduce the concepts of hubset, \hubtree\ and \hubtreewidth. These concepts are a generalization of the source set, \dagtree\ and \dagtreewidth\ respectively for directed graphs that are not acyclical. Using this new concepts, in \Sec{licl} we show that there is a relation between the \hubtreewidth\ of the $t$-fraternal extensions of $H$ and the $LICL$ of $H$, as given by the following lemma:
	
	\begin{restatable}{lemma}{lemmalicl} \label{lem:main_licl}
		Let $H$ be a pattern graph with $LICL(H) < 3(t+1)$, then for any $t$-fraternal extension $\vec{H}' \in \Frat{H, t}$ we have that $\htw(\vec{H}') = 1$.
	\end{restatable}
	
	Hence, because $H$ has an $LICL < 3(t+1)$ we will have that the \hubtreewidth\ of all the graphs in $\Frat{H}{t}$ is $1$.

	Finally, in \Sec{bressan} we show how to compute $\Hom{\MinFrat{\ExpG}{t}}{\vec{H}'}$ for each $\vec{H}' \in \Frat{H}{t}$. By replacing the \dagtree\ with the \hubtree\ we are able to adapt Bressan's algorithm \cite{Bressan2021} to work with labeled, weighted and directed graphs, giving the following lemma:
	
	\begin{restatable}{lemma}{lemmaBressan} \label{lem:main_bressan}
		Let $\vec{G}$ be a directed weighted and labeled graph with $n$ vertices and bounded outdegree and let $\vec{H}$ be a directed weighted and labeled graph with $\dtw(H) = 1$. There exists an algorithm that computes $\Hom{\vec{G}}{\vec{H}}$ in $O(n \log n)$ time.
	\end{restatable}
	
	Hence, we can compute $\Hom{\MinFrat{G}{t}}{\vec{H}'}$ in linear time for each graph $\vec{H}' \in \Frat{H}{t}$, as from \Lem{main_licl} we have that $\htw(\vec{H}')=1$ and from \Lem{main_grad} we have that $\Delta^+(\MinFrat{\ExpG}{t})$ is bounded. We can then aggregate the counts using \Lem{equivalence} to obtain the final homomorphism count. The whole process will take $O(n \log{n})$ time.

\end{proof}

\section{Labeled Graphs} \label{sec:label}

As we mentioned in the introduction, we can not use the fraternal extensions directly on $G$ and $H$. Fraternal extensions do not preserve the number of homomorphisms (see \Fig{homomorphism}). 
We define a pair of labeled graphs that can be obtained directly from $H$ and $G$ in linear time. Every homomorphism of the original graphs will translate into an injective homomorphism in the labeled graphs. We then will be able to do fraternal extensions on the labeled graphs while preserving the homomorphism count.

First, we define the labeled version of $H$, $\LabeledH$ which basically is $H$ but with every vertex labeled to itself:

\begin{definition}
	$\LabeledH$: Given a pattern graph $H=(V_H,E_H)$ we define the labeled graph $\LabeledH = (V_{\LabeledH},E_{\LabeledH}, L_{\LabeledH})$, where:
	\begin{itemize}
		\item $V_{\LabeledH} = V_H$
		\item $E_{\LabeledH} = E_H$
		\item $L_{\LabeledH}: V_{\LabeledH} \to V_{\LabeledH}$ is a labeling function such that $\forall v \in V_{\LabeledH},L_{\LabeledH}(v)=v$.
	\end{itemize}
\end{definition}

Now we define the graph $\ExpG$, this graph is obtained using the categorical product (\Def{categorical_product}) of $\LabeledH$ and $G$:

\begin{definition}
	$\ExpG$: Given a pattern graph $H=(V_H,E_H)$ and an input graph $G=(V_G, E_G)$, we define the labeled graph $\ExpG = (V_{\ExpG}, E_{\ExpG}, L_{\ExpG})$ as follows:
	\begin{itemize}
		\item $V_{\ExpG} = V_{H \times G}$
		\item $E_{\ExpG} = E_{H \times G}$
		\item $L_{\ExpG}: V_{\ExpG} \to V_{\LabeledH}, \text{ where } \forall \langle u,v \rangle \in V_{\ExpG}\ L(\langle u,v \rangle) = u$
	\end{itemize}
\end{definition}

Constructing $\LabeledH$ is trivial and takes constant time $O(k)$. We can construct $\ExpG$ efficiently as in the following claim:
\begin{claim} \label{clm:construct_G_1}
	$\ExpG$ has $O(n \cdot \kappa)$ vertices and $O(n \cdot \degen \cdot k^2)$, and can be constructed in $O(n\cdot \degen \cdot k^2)$ time.
\end{claim}
\begin{proof}
	We can construct the vertices and assign the labels in $O(n\cdot k)$ time. We will have at most $n \cdot \degen$ edges in $E_G$, and $k^2$ edges in $E_H$, hence we will take at most $O(n\cdot \degen \cdot k^2)$ to generate $E_{\ExpG}$. The total complexity will be $O(n\cdot \degen \cdot k^2)$.
\end{proof}

We can show that the number of homomorphisms from $H$ to $G$ is equivalent to the number of homomorphisms from $\LabeledH$ to $\ExpG$:

\begin{claim} \label{clm:hom_equivalence_1}
	\[
	\Hom{G}{H} = \Hom{\ExpG}{\LabeledH}
	\]
\end{claim}
\begin{proof}
	We show that there is a bijection between the homomorphisms from $H$ to $G$ and from $\LabeledH$ to $\ExpG$:
	
	\begin{itemize}
		\item Consider a homomorphism $\phi$ from $H$ to $G$, it will map the vertex $u$ to $\phi(u)$, we can create a homomorphism $\phi'$ from $\LabeledH$ to $\ExpG$ by mapping $u$ to the vertex $\langle u,\phi(u) \rangle$ in $V_{\ExpG}$ for all $u\in V_H$. If there is an edge in $H$ between $u_i, u_j$ we will have that there is also an edge between $\phi(u_i)$ and $\phi(u_j)$ in $G$, that implies by construction the existence of the arc $(\langle u_i,\phi(u_i)\rangle,\langle  u_j,\phi(u_j)\rangle)$ in $\ExpG$, and hence $\phi'$ will be a valid homomorphism.
		
		\item 	Similarly, given a homomorphism $\phi'$ from $\LabeledH$ to $\ExpG$ we can obtain a homomorphism $\phi$ from $H$ to $G$ by setting $\phi(u) = v : \phi'(u)=\langle u,v\rangle$. Again we need to show that this is a valid homomorphism: let $u_i,u_j$ be two vertices in $H$, we have $\phi(u_i) = v_i = v : \phi'(u_i)=\langle u_i, v\rangle$ and $\phi(u_j) = v_j = v : \phi'(u_j)=\langle u_j, v\rangle$.
		If there is an edge in $H$ between $u_i$ and $u_j$ we will have that there is also an edge between $\langle u_i,v_i \rangle$ and $\langle u_j,v_j \rangle$ in $\ExpG$, by construction that is only possible if there was also an edge between $v_i$ and $v_j$ in $G$ and hence $\phi$ preserves the edge.
	\end{itemize} 
\end{proof}

Therefore, we have proven that we can create these labeled graphs $\LabeledH$ and $\ExpG$ in linear time, and use them to count the homomorphisms instead. If we look at $\LabeledH$ we will have that $\LICL(\LabeledH) = \LICL(H)$ as the only difference between these graphs are the labels, which do not influence the structure of the graph. However, in order to be able to use these graphs we need to show a similar property for the input graph $G$ and its labeled version $\ExpG$, in this case, we need to be able to preserve the grad $\grad$, at least by a polynomial factor. In order to prove this we will use the following proposition from \cite{NeMe-book}:

\begin{proposition} [Prop. 4.6 \cite{NeMe-book}] \label{prop:grad}
	Let $G$ be a graph, let $p \geq 2$ be a positive integer and let $r$ be a half-integer. Then
	\[
	\topgrad_r(G \bullet K_p) \leq \max(2r(p-1)+1,p^2)\topgrad_r(G)+p-1
	\]
\end{proposition}

We can now prove the following lemma:

\begin{lemma} \label{lem:labeled_grad}
	Let $G$ be an input graph, and $H$ be a pattern graph with constant size $k$. For any constant $i$, $\grad_i(G)$ is bounded if and only if $\grad_i(\ExpG)$ is bounded.
\end{lemma}
\begin{proof}
	First, we have that $G \subseteq \ExpG$, hence $\grad_i(G)<\grad_i(\ExpG)$ and if $\grad_i(\ExpG)$ is bounded so will $\grad_i(G)$. We  show the other direction: note that $\ExpG$ will be a subgraph of $G \bullet K_k$:
	\[
	\ExpG = G \times H \subseteq G \times K_k = G \bullet \bar{K}_k \subseteq G \bullet K_k
	\]
	Where the second equality comes from \Fact{products}.
	
	Lastly, using \Prop{grad} we have that if $\topgrad_i(G)$ is bounded so will $\topgrad_i(G \bullet K_k)$, and hence so will $\topgrad_i(\ExpG)$. Combining this with \Fact{grad_topgrad} we get that $\grad_i(\ExpG)$ will be bounded if $\grad_i(G)$ is bounded.
\end{proof}
\begin{corollary}
	$G$ has bounded degeneracy if and only if $\ExpG$ has bounded degeneracy.
\end{corollary}

\section{Fraternal Extensions} \label{sec:fraternal}

In this section we formally introduce our augmentation procedure and define the concept of fraternal extensions and how to construct them efficiently. We will also show some properties of the fraternal extensions and prove the relation between the homomorphisms from $H$ to $G$ and the ones of their fraternal extensions.

\subsection{The Fraternal Extension Procedure}

We start by formally defining the fraternal extension of a graph:

\begin{definition} 
	\textbf{(Fraternal Extension)} Given a directed graph $\vec{G} = (V_{\vec{G}},E_{\vec{G}})$ we say that the directed weighted graph $\vec{G}'= (V_{\vec{G'}},E_{\vec{G'}}, W_{\vec{G'}})$ is a $t$-fraternal extension of $\vec{G}$ if $\vec{G}'$ forms a $t$-fraternity function and $E_{\vec{G}} = \{e \in E_{\vec{G'}} : W_{\vec{G'}}(e)=1\}$. 
	
	If $G$ is an undirected graph, we say that $\vec{G}'$ is a $t$-fraternal extension of $G$ if it is a $t$-fraternal extensions of some orientation $\vec{G}$ of $G$. The orientations of $G$ with unit weights are $1$-fraternal extensions of $G$.
\end{definition}

Abusing notation, given a directed weighted graph $\vec{G}'$ that forms $t$-fraternity function we say that the directed weighted graph $\vec{G}''$ is a $t'$-fraternal extension of $\vec{G}'$ if $\vec{G}''$ forms a $t'$-fraternity function and $\forall i \in (1,t), \vec{G}'_i = \vec{G}''_i$.

Note that all these definitions can be extended to labeled graphs.

We can construct fraternal extensions of a graph efficiently using a recursive procedure. The initial step is to give unit weights to every edge in the graph, each orientation of the resultant graph will be a different fraternal extension. Then, for $i \in [2,t]$, for every out-out wedge with combined weight is $i$ we add an undirected arc connecting it endpoints with weight $i$. Each orientation of the undirected edges will generate a distinct fraternal extension. The orientation of the undirected edges in each step will be different if we are performing this procedure on the pattern or on the input graph, as we will see in following subsections. We summarize the extension procedure in \Alg{extension}:

\begin{algorithm}
	\hspace*{\algorithmicindent} \textbf{Input:} \\
	\hspace*{6ex}-Directed graph $\vec{G}$ (should be a $(t-1)$-fraternal extension)\\
	\hspace*{6ex}-Integer $t$\\
	\hspace*{\algorithmicindent} \textbf{Output:}  \\
	\hspace*{6ex}-Set of edges $E^t$
	\begin{algorithmic}[1]
		\State Let $E^t=\emptyset$.
		\For{each out-out wedge $(u,v,w)$ in $\vec{G}$}
		\If{$W_{\vec{G}}(v,u) + W_{\vec{G}}(v,w) = t$ and $(u,w) \notin E_{\vec{G}}$ and $(w,u) \notin E_{\vec{G}}$}
		\State Add $(u,w)$ to $E^t$.
		\EndIf
		\EndFor
		\State Return $E^t$
	\end{algorithmic}
	\caption{\Extension($\vec{G}$,$t$)
	 }\label{alg:extension}
\end{algorithm}

We can prove that the fraternal extensions will connect the endpoints of out-out wedges.

\begin{claim} \label{clm:short-out-wedge}
	If a graph $\vec{G}'$ is a $t$-fraternal extension of a graph $\vec{G}$, then for every out-out wedge $(u,v,w)$ in $\vec{G}$ with $W_{\vec{G}}((v,u)) + W_{\vec{G}}(v,w) \leq t$ we have that there is an edge connecting $u$ and $w$ in $\vec{G}'$ of weight at most $t$.
\end{claim}
\begin{proof}
	If $\vec{G}'$ is a $t$-fraternal extension of a graph $\vec{G}$ then by the definition of fraternal extension we have that $\vec{G}'$ must form a $t$-fraternity function, thus there is a $t$-fraternity function $\omega: \mathcal{V} \times \mathcal{V} \to \mathbb{N} \cup \{\infty\}$ such that $\vec{G}' = G^\omega$. Let $(u,v,w)$ be an out-out wedge in $\vec{G}$ with $W_{\vec{G}}((v,u)) + W_{\vec{G}}(v,w) \leq t$.
	
	Abusing notation we will refer with $u,v,w$ to the correspondent elements in $\cV$. Because $\omega$ is a $t$-fraternity function we have by the definition of fraternity function (\Def{frat_function})that for $u$ and $w$, either:
	\begin{itemize}
		\item $\min(\omega(u,w),\omega(w,u))=1$, in which case there is an edge of weight $1$ connecting the two vertices.
		\item $\min(\omega(u,w),\omega(w,u)) = \min_{z\in \mathcal{V}\setminus\{u,w\}} \omega(z,u) + \omega(z, w)$, for $z=v$ we have $\omega(z,u) + \omega(z, w) = W_{\vec{G}}((v,u)) + W_{\vec{G}}(v,w) \leq t$, and hence we have that there will be an edge of weight at most $t$ connecting $u$ to $w$ or $w$ to $u$.
		\item $\min(\omega(x,y),\omega(y,x)) > t$ and $\min_{z\in \mathcal{V}\setminus\{x,y\}} \omega(z,x) + \omega(z, y)> t$, which is not possible as for $z= v$ we have $\omega(z,u) + \omega(z, w) = W_{\vec{G}}((v,u)) + W_{\vec{G}}(v,w) \leq t$.
	\end{itemize}
In all the cases we either reach to a contradiction or we prove the existence of the edge, completing the proof.
\end{proof}

\subsection{The Fraternal Extensions of $H$} \label{subsec:frat_h}

In the case of the pattern graph $H$, we are interested in generating all the possible $t$-fraternal extensions of its labeled version $\LabeledH$ at an specific depth $t$. We will use $\Frat{H}{t}$ to denote such collection of graphs. We define it as follows:

\begin{definition} 
	\textbf{$\Frat{H}{t}$} Given a pattern graph $H$ We call $\Frat{H}{t}$ to the collection of directed weighted and labeled graphs $\vec{H}'$ such that $\vec{H}'$ is a $t$-fraternal extension of any acyclical orientation of the labeled graph $\LabeledH$.
\end{definition}

\begin{figure}[H]
	\centering
	\includegraphics[width=\textwidth*3/4]{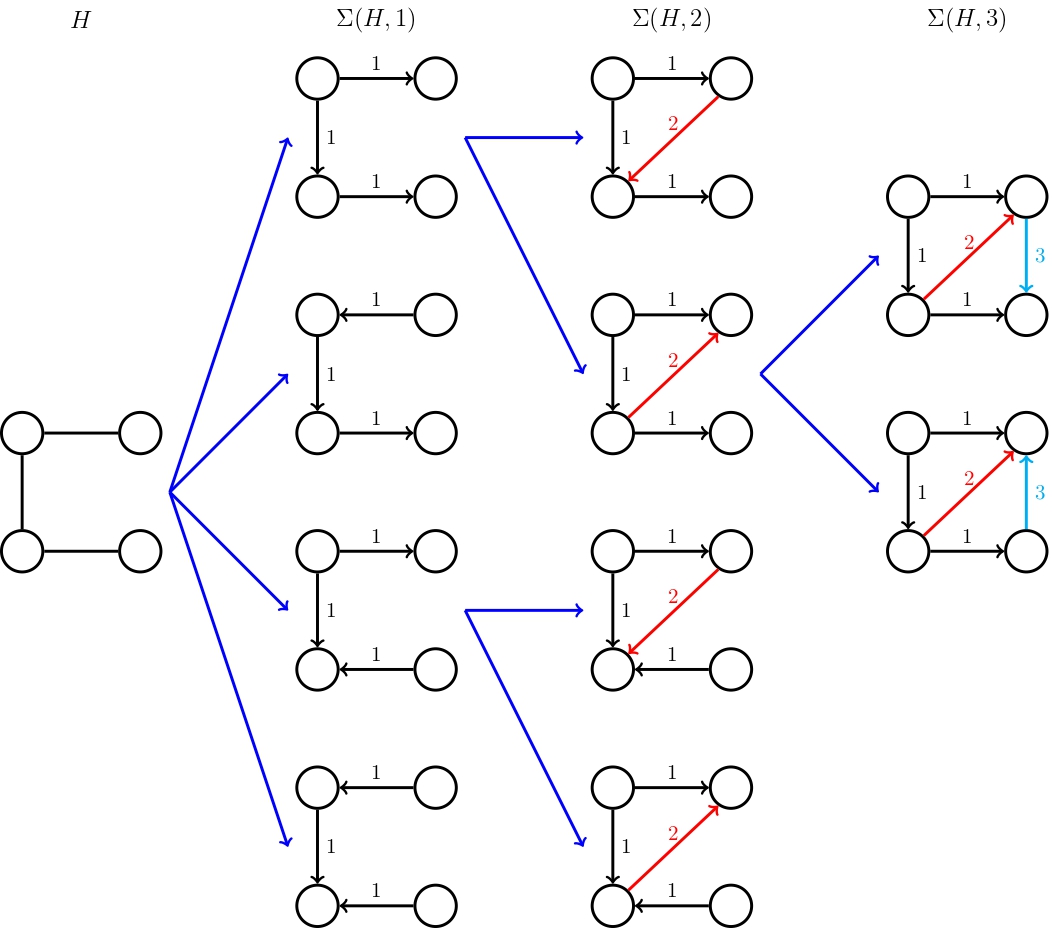}%
	\caption{An example of how the fraternal extensions of a graph $H$ are generated, the figure is not showing all the possible fraternal extensions. $\Frat{H}{1}$ corresponds to all the possible acyclical orientations of $H$. $\Frat{H}{2}$ adds the graphs obtained by adding edges between out-out wedges in the graphs of $\Frat{H}{1}$, this may generate some new out-out wedges that will be connected by edges in the graphs in $\Frat{H}{3}$. Note that a graph in $\Frat{H}{1}$ with no out-out wedges will also be part of $\Frat{H}{2}$ and $\Frat{H}{3}$.} 
	\label{fig:branching}
\end{figure}

Note that $\Frat{H}{1}$ corresponds exactly with the collection of acyclical orientations of $H$ (adding unit weights to the edges and trivial labels to the vertices). 

Let $\vec{H}_i$ be a $i$-fraternal extension of $H$. Abusing notation, we will use $\Frat{\vec{H}_i}{t}$ for $t>i$ to denote the set of $t$-fraternal extensions $\vec{H}'$ of $H$ such that $\vec{H}'$ is also a $t$-fraternal extension of $\vec{H}_i$.

Obtaining the collection of fraternal extensions of a graph can be seen as an iterative process where we apply the extension procedure defined in \Alg{extension} to all the fraternal extensions of the previous layer and then consider all the possible orientations of these new edges, this process is summarized in \Alg{Frat}. \Fig{branching} shows an example of this process applied to a simple pattern graph.

\begin{algorithm}
	\hspace*{\algorithmicindent} \textbf{Input:} \\
	\hspace*{6ex}-Labeled Graph $\LabeledH$\\
	\hspace*{6ex}-Integer $t$\\
	\hspace*{\algorithmicindent} \textbf{Output:}  \\
	\hspace*{6ex}-$\Frat{H}{t}$
	\begin{algorithmic}[1]
		\State Set $\Frat{H}{1} = \emptyset$
		\For{each acyclical orientation $\vec{H}$ of $\LabeledH$}
			\State Add $\vec{H}$ to $\Frat{H}{1}$ with unit weights.
		\EndFor
		\For{$i \in [2,t]$} 
			\State Set $\Frat{H}{i} = \emptyset$
			\For{$\vec{H} \in \Frat{H}{i-1}$}
				\State Let $E^i =$ \Extension($\vec{H},i)$
				\For{each orientation $\vec{E}^i$ of the edges $E^i$ with weights $i$}
					\State Add $\vec{H} \cup \vec{E}^i$ to $\Frat{H}{i}$
				\EndFor
			\EndFor
		\EndFor
		\State Return $\Frat{H}{t}$
	\end{algorithmic}
	\caption{\ComputeFrat($\LabeledH$, $t$)}\label{alg:Frat}
\end{algorithm}

\subsection{The Fraternal Extensions of $\ExpG$} \label{subsec:buildingG}

We now formally define the optimal acyclic $t$-fraternal extension of $\ExpG$:

\begin{definition} 
	\textbf{(Optimal Acyclic $t$-Fraternal Extension $\MinFrat{\ExpG}{t}$)} Given a graph $\ExpG$ we use $\MinFrat{\ExpG}{t}$ to denote an optimal acyclic $t$-fraternal extension of $\ExpG$. It is defined as follows:
	\begin{itemize}
		\item $\MinFrat{\ExpG}{t}$ is a $t$-fraternal extension of $\ExpG$.
		\item For every $i \in [1,t]$ we have that the edges ${\MinFrat{\ExpG}{t}}_i$ are oriented acyclically following the degeneracy orientation of the edges with weight $i$.
	\end{itemize}
\end{definition}
Note that for a graph $F$ there are multiple graphs that will follows those properties, as the degeneracy orientation of a set of undirected edges might not be unique. We can select any arbitrary degeneracy orientation whenever it is not unique.

We can show that there is a strong relation between the $\grad_{(t-1)/2}$ of a graph $\ExpG$ and the maximum outdegree of any of its optimal acyclic $t$-fraternal extensions:

\lemmagrad*
\begin{proof}
	We prove by induction on $t$:
	
	When $t=1$, $\MinFrat{\ExpG}{1}$ corresponds with the degeneracy orientation of $\ExpG$. Because $\ExpG$ has bounded $\grad_{(t-1)/2}$, we get that it also has bounded $\grad_{0}$ and hence, bounded degeneracy. Therefore, the degeneracy orientation of $\ExpG$ will have bounded outdegree $\Delta^+(\MinFrat{\ExpG}{1})$.
	
	Now, for the inductive step, assume that for any $k<t$, $\Delta^+(\MinFrat{\ExpG}{k})$ is bounded. We prove that $\Delta^+(\MinFrat{\ExpG}{k+1})$ will also be bounded. First, we bring the following lemma from \cite{NeMe-book}:
	
	\begin{lemma} [Lemma 7.6 in \cite{NeMe-book}] \label{lem:grad_to_delta}
		Let $\mathcal{V}$ be a finite set, let $ k\geq 1$ be an integer, let $\omega$ be a $k$-fraternity function and let $\ExpG=G_1^\omega$.
		
		There exists a $(k+1)$-fraternity function $\omega'$ such that:
		\[
		\forall(x,y)\in \mathcal{V}^2, \omega(x,y) \leq k \Longrightarrow \omega'(x,y) = \omega(x,y)
		\]
		\[
		\Delta^+_{k+1}(\omega') \leq \topgrad_{k/2}(\ExpG \bullet \overline{K}_{1+N_\omega(k+1)})
		\]
	\end{lemma}
	Here $N_\omega(k+1)$ is a real valued function defined in section 7.4 of \cite{NeMe-book} and its value is polynomial in $\Delta_1^+(\omega)+,\ldots,\Delta_{k}^+(\omega)$.
	
	Let $\omega_k$ be the fraternity function of $\MinFrat{\ExpG}{k}$. By our assumption $\Delta^+(\MinFrat{\ExpG}{k})$ is bounded, hence $\Delta_1^+(\omega)+,\ldots,\Delta_{k}^+(\omega)$ will also be bounded and so will $N = N_\omega(k+1)$. Applying \Lem{grad_to_delta} we know that there exists a function $\omega'$ such that $\forall(x,y)\in \mathcal{V}^2, \omega_k(x,y) \leq k \Longrightarrow \omega'(x,y) = \omega_k(x,y)$ and $\Delta^+_{k+1}(\omega') \leq \topgrad_{k/2}(\ExpG \bullet \overline{K}_{1+N})$.
	
	In other words, there exists a $k+1$-fraternal extension $\vec{\ExpG}' = G^{\omega'}$ of $\ExpG$ such that its outdegree is at most $\topgrad_{k/2}(\ExpG \bullet \overline{K}_{1+N_\omega(k+1)})$. We must show now that $\topgrad_{k/2}(\ExpG \bullet \overline{K}_{1+N_\omega(k+1)})$ is bounded: 
	
	Recall that $N = N_\omega(k+1)$ is bounded. We have that $k<t$ hence $k/2 \leq (t-1)/2$ and $\grad_{k/2}(\ExpG) \leq \grad_{(t-1)/2}(\ExpG)$. Therefore, because $\ExpG$ has bounded $\grad_{(t-1)/2}$, it will also have bounded $\grad_{k/2}$. Now, applying \Prop{grad} we get that $\topgrad_{k/2}(\ExpG \bullet \overline{K}_{1+N})$ will also be bounded as $N$ is a constant, and we get that:
	\begin{align*}
		\Delta^+(\vec{\ExpG}')_{k+1} = \Delta^+(G^{\omega'})_{k+1} = \Delta^+_{k+1} \leq \topgrad_{k/2}(\ExpG \bullet \overline{K}_{1+N})
	\end{align*}
	Therefore $\Delta^+(\vec{\ExpG}')_{k+1}$ will be bounded.
	
	However is possible that $\omega'$ is orienting the edges of the $k+1$ layer in an orientation that is not acyclical, we can show that the max outdegree of the degeneracy orientation is as most a factor of $2$ with respect to the minimum max out degree of all cyclical orientations:
	
	\begin{claim}
		Given an undirected graph $G$, let $\optimalOrientation$ be the orientation of $G$ with minimal $\Delta^+$ and let $\degenOrientation$ be the degeneracy orientation of $G$:
		\[
		\Delta^+(G^\degen) \leq 2\Delta^+(G^*)
		\]
	\end{claim}
	\begin{proof}
		$\Delta^+(\degenOrientation)$ is equal to the degeneracy $\degen$ of $G$. There exists a subgraph $G'$ in $G$ such that has minimum degree $\degen = \Delta^+(\degenOrientation)$. Let $n'$ be the number of vertices in $G'$ and $m'$ the number of edges. We will have that $m' \geq \frac{n' \cdot \Delta^+(\degenOrientation)}{2}$.
		
		If we divide all the $m'$ edges equally so that the maximum outdegree in $G'$ is minimized we will have that every vertex has an outdegree of at least $\frac{\Delta^+(\degenOrientation)}{2}$. Hence any orientation of the edges of $G'$ will have an outdegree of at least $\frac{\Delta^+(\degenOrientation)}{2}$. Therefore, $\Delta^+(\optimalOrientation) \geq \frac{\Delta^+(\degenOrientation)}{2}$.
	\end{proof}
	
	Therefore $\Delta^+_{k+1}(\MinFrat{\ExpG}{k+1}) \leq 2 \Delta^+(\vec{\ExpG}')_{k+1}$, and hence, it is bounded. Because for all the other layers $\MinFrat{\ExpG}{k}$ and $\MinFrat{\ExpG}{k+1}$ are identical and $\MinFrat{\ExpG}{k}$ has bounded outdegree, we get that $\Delta^+(\MinFrat{\ExpG}{k+1})$ is bounded.
	
\end{proof}

Additionally, If $\ExpG$ belongs to the class of graphs with bounded rank $(t-1)/2$ grad. we can then compute $\MinFrat{\ExpG}{t}$ efficiently. The process is summarized in \Alg{G_construct}, and basically consist of perform multiple iterations of the extension procedure, orienting the edges by the degeneracy orientation each time. We get the following lemma:

\buildingG*

\begin{proof}
	First, because $\ExpG$ has bounded $\grad_{(t-1)/2}$ it will also have bounded degeneracy $\kappa$. We can then construct $\MinFrat{\ExpG}{1}$ by orienting $\ExpG$ acyclically using the degeneracy orientation, by \Fact{degen_orientation} this will take linear time on the number of vertices and edges of $\ExpG$, which is still $O(n)$.
	
 	By \Lem{grad_to_delta} we will have that for all $k\leq t$ the graph $\MinFrat{\ExpG}{k}$ has bounded $\Delta^+$. We show an inductive process where given $\MinFrat{\ExpG}{k}$ we can compute $\MinFrat{\ExpG}{k+1}$ in linear time:
	
	\begin{itemize}
		\item First, compute all the out-out wedges of $\MinFrat{\ExpG}{k}$. This can be done in linear time as the outdegree of every vertex is bounded. The number of out-out wedges will also be linear in $n$.
		
		\item For every out-out wedge $(u,v,w)$, if the sum of the weights of the two edges $(v,u)$ and $(v,w)$ is $k+1$ then create an edge connecting $(u,w)$ with weight $k+1$. This process takes linear time in the number of out-out wedges, which is again, linear in $n$. And there will be at most $O(n)$ edges in the $k+1$ layer.
		
		\item Finally, orient the newly created edges using the degeneracy orientation when considering only the new edges. This again can be done in $O(n)$ time.
	\end{itemize}
	We complete the proof by showing that the resultant graph $\vec{\ExpG}'$ is a valid $\MinFrat{\ExpG}{k+1}$. First, we can see that every layer in the graph is oriented by the degeneracy orientation. That was true for the layer $1$ to $k$ as we started our construction with $\MinFrat{\ExpG}{k}$ and we haven't added any additional edge with a weight $\leq k$. It is also true for the $k+1$ layer as the last step of our construction orients the layer using the degeneracy orientation.
	
	Now we just need to show that $\vec{\ExpG}'$ is a $k+1$-fraternal extension of $\ExpG$. Let $\cV = V_{\vec{\ExpG}'}$ and $\omega: \cV \times \cV \to \mathbb{N}$ be a function such that $G^\omega =\vec{\ExpG}'$. We must show that $\omega$ forms a $k+1$-fraternity function. We prove by contradiction, if $\omega$ is not a $k+1$ fraternity function, then by \Def{frat_function} we have that there must exist a pair of vertices in $V_{\vec{\ExpG}'}$ such that the equivalent nodes $u,v$ in $\cV$ do not meet any of the following conditions:
	
	\begin{enumerate}
		\item $\min(\omega(u,v),\omega(v,u))=1$
		\item $\min(\omega(u,v),\omega(v,u)) = \min_{w\in \mathcal{V}\setminus\{u,v\}} \omega(w,u) + \omega(w, v)$
		\item $\min(\omega(u,v),\omega(v,u)) > k+1$ and $\min_{w\in \mathcal{V}\setminus\{u,v\}} \omega(w,u) + \omega(w, v)> k+1$.
	\end{enumerate}
	
	If $u$ and $v$ do not meet the first condition then we have that $\min(\omega(u,v),\omega(v,u))>1$ (recall that in a fraternity function a missing edge is considered as $\infty$). If $\min(\omega(u,v),\omega(v,u))<k+1$ then $u,v$ was an edge in $\MinFrat{\ExpG}{k}$ which implies that it will meet the second condition, as $\MinFrat{\ExpG}{k}$ was a valid $k$-fraternity extension. Hence, we have that $\min(\omega(u,v),\omega(v,u))$ is either $k+1$ or $>k+1$:
	\begin{itemize}
		\item In the first case $\min(\omega(u,v),\omega(v,u)) = k+1$: because the second condition is not true we will have that $\min_{w\in \mathcal{V}\setminus\{u,v\}} \omega(w,u) + \omega(w, v) \neq k+1$. If it is greater than $k+1$ then our procedure would not have generated an edge connecting $u,v$ with weight $k+1$. Otherwise we have that is lower than $k+1$ we have that $\min_{w\in \mathcal{V}\setminus\{u,v\}} \leq k$, but then in $\MinFrat{\ExpG}{t}$ we will have an edge connecting $u$ and $v$ with weight $k$ or it would not be a valid fraternal extension. Both cases reach to a contradiction.
		
		\item In the second case we have that $\min(\omega(u,v),\omega(v,u))>k+1$: Then because the third condition is false, we will have that $\min_{w\in \mathcal{V}\setminus\{u,v\}} \omega(w,u) + \omega(w, v)\leq k+1$. But in that case an out-out wedge with weight at most $k+1$ would connect $u$ and $v$, and hence our procedure would have generated an edge connecting $u$ and $v$ with weight at most $k+1$, again reaching a contradiction.
	\end{itemize}
\end{proof}

\begin{algorithm}
	\hspace*{\algorithmicindent} \textbf{Input:} \\
	\hspace*{6ex}-Labeled Graph $\ExpG$\\
	\hspace*{6ex}-Integer $t$\\
	\hspace*{\algorithmicindent} \textbf{Output:}  \\
	\hspace*{6ex}-$\MinFrat{\ExpG}{t}$
	\begin{algorithmic}[1]
		\State Let $\MinFrat{\ExpG}{1}$ be the degeneracy orientation of $\ExpG$ with unit weights.
		\For {$k \in [2,t]$ }
			\State Let $E^i =$ \Extension($\MinFrat{\ExpG}{i-1},i)$
			\State Let $\vec{E}^i$ be the degeneracy orientation of the edges in $E^i$ with weight $i$.
			\State Set $\MinFrat{\ExpG}{i} = \MinFrat{\ExpG}{i-1} \cup \vec{E}^i$
		\EndFor
		\State Return $\MinFrat{\ExpG}{t}$
	\end{algorithmic}
	\caption{\ComputeMinFrat($F$,$t$)} \label{alg:G_construct}
\end{algorithm}

\subsection{Equivalence of Homomorphisms} \label{subsec:equivalence}

In this section we prove the equivalence between the homomorphisms of the original graphs and the fraternal extensions. This is given by \Lem{equivalence}, that we restate:

\equivalence*

\begin{proof}

	Let $\Phi(\LabeledH,\ExpG)$ be the set of homomorphisms from $\LabeledH$ to $\ExpG$. For every $\vec{H}'\in \Frat{H}{t}$, let $\Phi(\vec{H}',\MinFrat{\ExpG}{t})$ be the set of homomorphisms from $\vec{H}'$ to $\MinFrat{\ExpG}{t}$. We can see that each of these sets are disjoint:
	
	\begin{claim} 
		Let $\vec{H}',\vec{H}'' \in \Frat{H}{t}$ be two distinct $t$-fraternal extensions of $H$:
		\begin{align*}
			\Phi(\vec{H}',\MinFrat{\ExpG}{t})\cap \Phi(\vec{H}'',\MinFrat{\ExpG}{t}) = \emptyset.
		\end{align*}

	\end{claim}
	\begin{proof}
		First, we show that if $\vec{H}'$ and $\vec{H}''$ are distinct $t$-fraternal extensions of $H$ then there must exist an edge $e \in E_{\vec{H}'}$ such that it is reversed in $\vec{H}''$. We prove it by contradiction: Assume there is no such edge, we show that then $\vec{H}' = \vec{H}''$. We do induction in the depth of the fraternal extensions:
		
		\begin{itemize}
			\item For the base case, $\vec{H'}$ and $\vec{H}''$ are both $1$-fraternal extensions and therefore they will correspond to different orientations of the edges in $H$, if they don't differ in any edge then $\vec{H}'$ and $\vec{H}''$ will correspond to the exact same orientation.
			
			\item For the inductive step, assume for some $k<t$ that if there is no reversed edge in two $k$-fraternal extensions $\vec{H}'$,$\vec{H}''$ of $H$, then $\vec{H}'=\vec{H}''$. We show that same holds for $k+1$-fraternal extensions. By the assumption we know that two $k$-fraternal extensions that do not differ in any edge will correspond to the same graph. Hence, two $k+1$-fraternal extension $\vec{H}',\vec{H}'' \in \Frat{H}{k+1}$ will have the same edges (with the exact same orientations) up to the layer $k$, and therefore the edges in the layer $k+1$ must be the same. Because we are assuming that there are not reversed edges, they will also have the same orientation. Therefore $\vec{H}'=\vec{H}''$.
		\end{itemize}
		
		Hence, there exists an edge $e$ that have different orientations in $\vec{H}'$ and $\vec{H}''$. Let $u,v$ be the endpoints of the edge $e$, the arc $(u,v)$ belongs to $\vec{H}'$ and the arc $(v,u)$ to $\vec{H}''$. We show that no homomorphism $\phi$ can be both in $\Phi(\vec{H}',\MinFrat{\ExpG}{t})$ and in $\Phi(\vec{H}'',\MinFrat{\ExpG}{t})$. Consider the vertices $\phi(u)$ and $\phi(v)$ of $\MinFrat{\ExpG}{t}$, we can have either an arc from $\phi(u)$ to $\phi(v)$, from $\phi(v)$ to $\phi(u)$, or none. In order for $\phi$ to be a valid homomorphism of $\vec{H}'$, we will need the directed arc $(\phi(u), \phi(v))$ to be in $\MinFrat{\ExpG}{t}$, similarly for $\vec{H}''$ we will need the directed arc $(\phi(v),\phi(u))$ to be in $\MinFrat{\ExpG}{t}$. A fraternal extension can not contain two opposite edges connecting the same two vertices. Therefore $\phi$ can not be a homomorphism of $\vec{H}'$ and $\vec{H}''$ at the same time.
	\end{proof}
	
	Now we just need to show that:
	
	\begin{align*}
		\Phi(\LabeledH,\ExpG) = \bigcup_{\vec{H}' \in \Frat{H}{t}}\Phi(\vec{H}',\MinFrat{\ExpG}{t})
	\end{align*}
	
	We start by proving that for every $\vec{H}' \in \Frat{H}{t}$, if $\phi$ is a homomorphism from $\vec{H}'$ to $\MinFrat{\ExpG}{t}$ then it will also be a valid homomorphism from $\LabeledH$ to $\ExpG$: Let $\vec{H}' \in \Frat{H}{t}$ be a $t$-fraternal extension of $H$ and $\phi$ a homomorphism from $\vec{H}'$ to $\MinFrat{\ExpG}{t}$, $\phi$ must be an injective mapping, hence every vertex in $\vec{H}'$ is mapped to a distinct vertex in $\MinFrat{\ExpG}{t}$. 
	
	We show that $\phi$ is also a valid homomorphism from $\LabeledH$ to $\ExpG$, consider the edge $(u,v) \in \LabeledH$, we must show that the edge $(\phi(u),\phi(v))$ is present in $\ExpG$. Because $\vec{H}'$ is a fraternal extension of $\LabeledH$, we will have that they share the edges of weight $1$, hence either the arc $(u,v)$ or the arc $(v,u)$ will be present in $\vec{H}'$ (we can assume without loss of generality that it is oriented from $u$ to $v$) with unit weight $W_{\vec{H}'}((u,v))=1$. Because $\phi$ is a homomorphism from $\vec{H}'$ to $\MinFrat{\ExpG}{t}$ we will have that the edge $(\phi(u),\phi(v))$ must be present in $\MinFrat{\ExpG}{t}$ with weight $W_{\MinFrat{\ExpG}{t}}((u,v)) = W_{\vec{H}'}((u,v))=1$. For $\MinFrat{\ExpG}{t}$ to have an edge with weight $1$, such edge must be also in $\ExpG$, and hence $(\phi(u),\phi(v))$ is present in $E_{\ExpG}$.
	
	Now we prove that if $\phi$ is a homomorphism from $\LabeledH$ to $\ExpG$ then there exists a fraternal extension $\vec{H}' \in \Frat{H}{t}$ such that $\phi$ is a valid homomorphism from $\vec{H}'$ to $\MinFrat{\ExpG}{t}$: Let $\phi$ be a homomorphism from $\LabeledH$ to $\ExpG$, we show that we can construct a $t$-fraternal extension of $\LabeledH$ such that $\phi$ is a valid homomorphism from it to $\MinFrat{\ExpG}{t}$. We use induction on $t$:
	
	\begin{itemize}
		\item For the base case $k=1$, we can orient every edge $(u,v)$ in $\LabeledH$ so it matches the orientation of the edge $(\phi(u),\phi(v))$ in  $\MinFrat{\ExpG}{1}$. This orientation will be acyclic and therefore the resultant graph will be in $\Frat{H,1}$.
		
		\item For the inductive step, we assume that for $k<t$ exists a $k$-fraternal extension $\vec{H}^k$ of $H$ where $\phi$ is a valid homomorphism from $\vec{H}^k$ to $\MinFrat{\ExpG}{k}$. We prove that it also holds for $k+1$: First, let $\vec{H}'$ be a $k+1$-fraternal extension of $\vec{H}^k$, any edge $(u,v)$ in $\vec{H}'$ with $W_{\vec{H}'}((u,v))<t$ must have a correspondent edge $(\phi(u),\phi(v))$ in $\MinFrat{\ExpG}{k}$ and hence in $\MinFrat{\ExpG}{k+1}$ with $W_{\MinFrat{\ExpG}{k+1}}((\phi(u),\phi(v))) < W_{\vec{H}'}((u,v))$. 
		
		Hence we only need to verify that the edges $e$ in $\vec{H}'$ with $W_{\vec{H}'}(e)=k+1$ are also present in $\MinFrat{\ExpG}{t}$. Consider the edge $e=(u,v)$ in $\vec{H}'$ with $W_{\vec{H}'}(e)=k+1$, we need to show that $\MinFrat{\ExpG}{t}$ contains either $(\phi(u),\phi(v))$ or $(\phi(v),\phi(u))$ as the edges in the last layer of $\vec{H}'$ can be oriented arbitrarily and still will be a valid $k+1$-fraternal extension of $\vec{H}^k$. If such edge $e$ exists, then there must exist a vertex $w$ in $\vec{H}'$ such that there is an out-out wedge $(u,w,v)$ in $\vec{H}^k$ with total weight $k+1$, by the assumption, we will have that there is an out-out wedge $(\phi(u),\phi(w),\phi(v))$ in $\MinFrat{\ExpG}{k}$ with total weight $\leq k+1$. Hence $\MinFrat{\ExpG}{t}$ must include either $(\phi(u),\phi(v))$ or $(\phi(v),\phi(u))$ with weight $\leq k+1$.
		
	\end{itemize}

\end{proof}

\section{The Hub-Set} \label{sec:hub-set}

As we mentioned before, the fraternal extensions of an acyclic graph might no longer be acyclic. Bressan's algorithm for counting homomorphism requires of directed acyclic graphs, as it relays on the definitions of \dagtree\ and \dagtreewidth{}. We will generalize these definitions and extend Bressan's algorithm to directed graphs that are not necessarily acyclic. For that purpose we introduce the concept of hubset of a directed graph:

\begin{definition}
	Let $\vec{H}$ be a directed graph, a hubset of $\vec{H}$ is any subset of vertices $\HubSet \subseteq V_{\vec{H}}$ such that:
	\begin{itemize}
		\item For each pair $s,s'\in \HubSet$ with $s\neq s'$ there is no directed path connecting $s$ to $s'$ (or vice versa).
		\item For each vertex $v \in V_{\vec{H}} \setminus \HubSet$, there exists a vertex $s \in \HubSet$ such that there is a directed path connecting $s$ to $v$.
	\end{itemize}
\end{definition}

We will use $\HubSet(\vec{H})$ to denote any hubset of $\vec{H}$. Note that when $\vec{H}$ is acyclic the hubset of $\vec{H}$ is unique and corresponds exactly with the source set. Furthermore, this applies to any fraternal extension of a DAG as we can see in the following claim:

\begin{claim}
	Let $\vec{H}'$ be a fraternal extension of the DAG $\vec{H}$. $\vec{H}$ has an unique hubset and $\HubSet(\vec{H}') = S(\vec{H})$.
\end{claim}
\begin{proof}
	Note that the edges of $\vec{H}$ are a subset of the edges of $\vec{H}'$. Hence, the second condition for the hubset is automatically satisfied as every vertex in $\vec{H}$ is reachable from at least one source in $S(\vec{H})$. For the first condition suffices to observe that the indegree of any of the sources will be $0$ in all the fraternal extensions of $\vec{H}$. We can prove it by contradiction, assume that there is a fraternal extension that adds an edge incident to the source $s \in S(\vec{H})$, then we will have that $s$ was one of the endpoints of an out-out wedge. That is not possible as $s$ is a source and hence its initial indegree in $\vec{H}$ was $0$.
	
	Now we prove that the hubset is unique, note that by the previous argument every source of $\vec{H}$ will still be source in $\vec{H}'$. Hence all the sources must be included in the hubset in order to satisfy the second condition. Adding any additional vertex in the hubset is also not possible, as all the vertices are reachable from at least one source and we would violate the first condition.
\end{proof}

Note that every graph $\Frat{H}{1}$ is a DAG, therefore the previous claim will apply to every graph in $\Frat{H}{t}$ for every $t>0$.

We can now define a new type of decomposition of a graph based on the hubset. This new decomposition is just a generalization of Bressan's \dagtree\ for directed graphs:

\begin{definition} [\hubtree]
	Let $\vec{H}$ be a directed graph with hubset $\HubSet = \HubSet(\vec{H})$. A \hubtree of $\vec{H}$ is a rooted tree $\cT=(\cB,\cE)$ with the following properties:
	\begin{enumerate}
		\item Each node $B\in \cB$ is a subset of $\HubSet$ of $\vec{H}$, $B \subseteq \HubSet$.
		\item $\bigcup_{B \in \cB} B = \HubSet$.
		\item For all $B,B_1,B_2 \in \mathcal{B}$, if $B$ is on the unique path between $B_1$ and $B_2$ in $\mathcal{T}$, then we have $\Reachable_{\vec{H}}(B_1)\cap \Reachable_{\vec{H}}(B_2) \subseteq \Reachable_{\vec{H}}(B)$.
	\end{enumerate}
\end{definition}

We similarly define the \hubtreewidth\ of a graph: 

\begin{definition} [\hubtreewidth ($\htw$)]
	The \hubtreewidth\ of a \hubtree\ $\cT = (\cB, \cE)$ is defined as the maximum size over all the nodes of $\cT$:
	\begin{align*}
		\htw(\cT) = \max_{B \in \cB} |B|
	\end{align*}
	We also use $\htw(\vec{H})$ to refer to the \hubtreewidth\ of the directed graph $\vec{H}$, which is the minimum $\htw(\cT)$ over all possible \hubtree\ of $\vec{H}$.
\end{definition}

Note that when $\vec{H}$ is acyclic the definitions for \dagtree\ and \hubtree\ are equivalent. Similarly we will have that the \dagtreewidth\ and \hubtreewidth\ are the same, that is why we will refer to both of them using $\htw$.

\section{The \hubtreewidth\ of fraternal extensions and the LICL} \label{sec:licl}

As we can see in \Fig{licl}, the fraternal extensions do not necessarily reduce the $\LICL$ of the pattern graphs, as new cycles can be formed with the new edges in the extensions. However, there is a clear relation between the $\LICL$ of the original graph and the \hubtreewidth\ of the fraternal extensions. We will be proving such relation in this section, given by the following lemma:

\lemmalicl*

\subsection{Main Technical Lemma}

In this subsection we prove the main technical lemma of this paper, which will allow us to prove the relation between fraternal extensions and dag-treewidth. First we will define a long out-out wedge:

\begin{definition}
	A long out-out wedge is a graph form by the union of two directed paths of any length as the result of combining their sources. \Fig{outwedge} shows an example of a long out-out wedge.
\end{definition}

We now prove the following claims that will be used in the main lemma of this subsection:

\begin{claim} \label{clm:out-wedge}
	Let $u,v$ be the endpoints of a long out-out wedge with total weight $w$ and $l$ edges in some fraternal extension $\vec{H}_i$ of $H$, then for all $t \geq w$, for all $\vec{H}' \in \Frat{\vec{H}_i}{t}$ there is a direct path connecting either $u$ to $v$ or $v$ to $u$ using only the vertices in the long out-out wedge.
\end{claim}
\begin{proof}
	We can prove by induction on the number of edges of the long out-out wedge:
	
	The base case is when $l=2$, this is simply a standard out-out wedge with where the two edges have a combined weight of $w$. By \Clm{short-out-wedge} we have that any $t$-fraternal extension of $\vec{H}_i$ at depth $t \geq w$ will add an edge connecting $u,v$ if it was not already present. Hence, we will either have a path from $u$ to $v$ or $v$ to $u$.
	
	Now we show the inductive step: Assume that the claim holds for $l=k$, we will prove that it also holds for $l=k+1$. 
	
	Let $s$ be the source of the long out-out wedge with length $k+1$, $s$ will be forming a normal out-out wedge with a vertex $u'$ in the $s-u$ path and a vertex $v'$ in the $s-v$ path, both $(s,u')$ and $(s,v')$ edges will have a combined weight strictly less than $w$, thus for some $t'<w$ by \Clm{short-out-wedge} any fraternal extension of $H_i$ at level $t'$ will have an edge $(u',v')$ (we can assume without loss of generality that it will go from $u'$ to $v'$).
	
	If $u'=u$ then we actually a directed path from $u$ to $v$, otherwise, we have a long out-out wedge where the total weight is still at most $w$, but with the source at $u'$, with one less edge, hence $l=k$. Using the assumption of the inductive step we know that there will be a path either from $u$ to $v$ or from $v$ to $u$ in the $t$-fraternal extension.
	
\end{proof}

In \Fig{outwedge} we can see an example of the long out-out wedge construction.

\begin{figure}
	\centering
	\includegraphics[width=\textwidth*3/4]{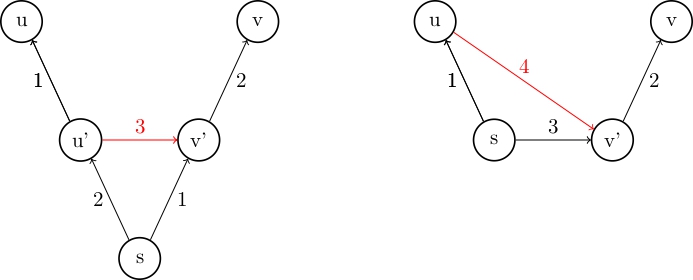}%
	\caption{An example of connectivity between the end-points of a long out-out wedge, we can see how $u'$ becomes the source of a new long out-out wedge after adding the edge of weight $3$. When adding the edge of weight $4$ in the right picture we end up with a direct path between $u$ and $v$.}
	\label{fig:outwedge}
\end{figure}

We now prove the following:

\begin{claim} \label{clm:reachability-in-paths}
	Let $P$ be an induced undirected path of length $t$ from $u$ to $v$ in $H$, then for any $t$-fraternal extension $\vec{H} \in \Frat{H}{t}$: There is an undirected path $P'$ connecting $u$ and $v$ (ignoring edge directions) using only vertices in $P$ with total weight at most $t$, such that $P'$ either:
	\begin{itemize}
		\item Case 1: Forms a direct path from $u$ to $v$.
		\item Case 2: Forms a direct path from $v$ to $u$.
		\item Case 3: Contains a vertex $s$ such that there is a direct path from both $u$ and $v$ to $s$ in $P'$ (a long in-in wedge).
	\end{itemize}
\end{claim}
\begin{proof}
	We can prove by induction on $t$:
	
	Our base case is when $t=1$, we have that $u$ and $v$ are connected by one edge in $H$, in every orientation of $H$ we will have that either the edge goes from $u$ to $v$ (Case 1) or from $v$ to $u$ (Case 2).
	
	For the inductive step we assume that the claim holds for all $t \leq k$, we will show that it also holds for $t= k+1$: 
	
	Let $u$ and $v$ be two vertices in $H$ inducing the path $P$ of length $k+1$. Let $u'$ be the vertex adjacent to $u$ in $P$, and let $e$ be the edge connecting $u$ and $u'$. Then $u'$ and $v$ form an induced path of length $k$ using the vertices of $P$. 
	
	Using the assumption of the inductive step, we have that for all the $k$-fraternal extensions $\vec{H}_k \in \Frat{H}{k}$ we will have that there is a path $P''$ connecting $u'$ and $v$ (ignoring directions) using only vertices in $P$ with total weight at most $k$, following one of the three cases. We will prove that all the cases lead to the construction of a path $P'$ from $u$ to $v$ for all $\vec{H}_{k+1}\in \Frat{\vec{H}_k}{k+1}$ following one of the three conditions:
	\begin{itemize}
		\item Case 1: We have that $P''$ is a direct path from $u'$ to $v$ with weight at most $k$. We have two possibilities depending on the orientation of $e$:
		\begin{enumerate}[(a)]
			\item $e$ is oriented from $u$ to $u'$: then $e \cup P''$ forms a direct path from $u$ to $v$ with total weight at most $k+1$.
			\item $e$ is oriented from $u'$ to $u$: then, in $\vec{H}_{k}$, $e \cup P''$ forms a long out-out wedge with the source at $u'$ with at most $k+1$ edges and at most $k+1$ total weight. Hence by \Clm{out-wedge} we have that there will be a direct path in all $\vec{H}_{k+1}\in \Frat{\vec{H}_k}{k+1}$ from $u$ to $v$ or from $v$ to $u$.
		\end{enumerate}
		\item Case 2: We have that $P''$ is a direct path from $v$ to $u'$ with weight at most $k$. We have two possibilities depending on the orientation of $e$:
		\begin{enumerate}[(a)]
			\item $e$ is oriented from $u$ to $u'$: then we have that $u'$ is reachable from both $u$ and $v$ in $\vec{H}_{k}$.
			\item $e$ is oriented from $u'$ to $u$: then $e \cup P''$ forms a direct path from $v$ to $u$ with total weight at most $k+1$. 
		\end{enumerate}
		\item Case 3: We have that $P''$ has a vertex $s$ such that there is a direct path from both $u'$ and $v$ to $s$ in $P''$. We have two possibilities depending on the orientation of $e$:
		\begin{enumerate}[(a)]
			\item $e$ is oriented from $u$ to $u'$: then we have a direct path from $u$ to $s$, and $e \cup P''$ form a long in-in wedge ending in $s$.
			\item $e$ is oriented from $u'$ to $u$: then we can see how in $\vec{H}_k$, $u$ and $s$ form a long out-out wedge in $e\cup P''$ centered in $u'$ with less than $k$ edges and less than $k$ weight. Hence by \Clm{out-wedge} we have that there will be a direct path in all $\vec{H}_{k+1}\in \Frat{\vec{H}_k}{k+1}$ from $u$ to $s$ or from $s$ to $u$. In the first case, $s$ will be reachable from both $u$ and $v$, and in the second case we have a direct path from $v$ to $u$.
		\end{enumerate}
	\end{itemize}
	
	As we can see, every possibility lead to one of the three cases in the claim. In \Fig{paths} there is a depiction of all the cases.
	
	\begin{figure}
		\includegraphics[width=\textwidth]{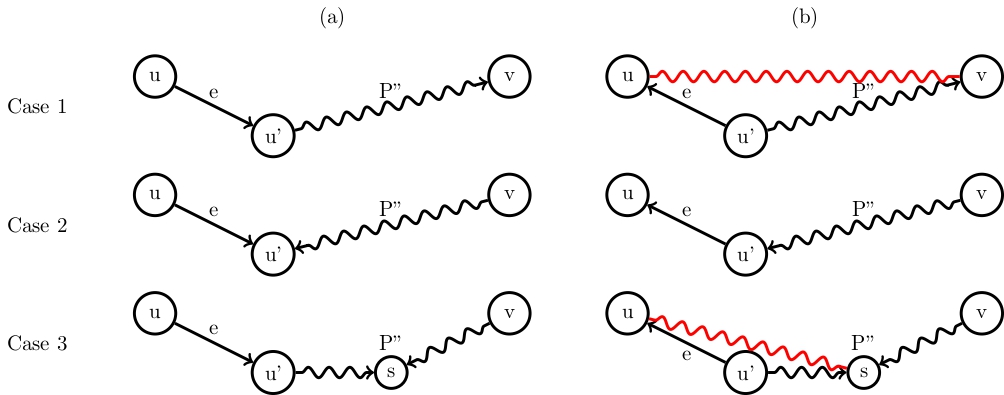}%
		\caption{The $6$ possible scenarios in the proof of \Clm{reachability-in-paths} depending on the orientations of $e$ and $P''$. The red edges represent the cases where a long out-out wedge appears, we know the endpoints of the long out-out wedges will connect because of \Clm{out-wedge}}
		\label{fig:paths}
	\end{figure}
\end{proof}

Before presenting the main technical lemma, we bring the definition of Unique reachability graph from \cite{Bera2021}. However this definition was created for directed acyclical graphs as it uses the sources of the graph. We adapt it to use the hubset instead:

\begin{definition} [Unique reachability graph] \label{def:urg}
	Let $\vec{H}$ be a directed graph with hubset $\HubSet = \HubSet(\vec{H})$ and $\HubSet_p \subseteq \HubSet$ be a subset of the hubset. We define a unique reachability graph $UR_{\HubSet_p}(\HubSet_p,E_{\HubSet_p})$ on the vertex set $\HubSet_p$, and the edge set $E_{\HubSet_p}$ such that there exists an edge $e = \{s_1,s_2\} \in E_{\HubSet_p}$, for $s_1,s_2 \in \HubSet_p$ if and only if the set $(\Reachable_{\vec{H}}(s_1) \cap \Reachable_{\vec{H}}(s_2))\setminus \Reachable_{\vec{H}}(\HubSet_p\setminus\{s_1,s_2\})$ is non-empty.
\end{definition}

We can now finally introduce the main technical lemma:

\begin{lemma} \label{lem:licl_to_acyclic}
	Given a pattern graph $H$, if $\LICL(H) < 3(t+1)$ then for any graph $\vec{H}' \in \Frat{H}{t}$ we have that $\UR_{\vec{H}'}$ is acyclical.
\end{lemma}
\begin{proof}
	We prove a stronger statement, let the graph $\vec{H}' \in \Frat{H}{t}$ be a $t$-fraternal extension of $H$. If $\UR_{\vec{H}'}$ contains a cycle of length $l$, then $H$ will contain an induced cycle of length at least $(t+1)\cdot l$.
	
	Consider the subset ${\HubSet}'(\vec{H}')$ of $l$ vertices from the hubset $\HubSet(\vec{H})$ forming the cycle in $\UR_{\vec{H}'}$, we have ${\HubSet}'(\vec{H}') \subseteq \HubSet(\vec{H}')$. We can enumerate them as $s_1, s_2, \ldots, s_l$, where for $i\in [1,l]$ we have that $s_i$ and $s_{i+1}$ (with $s_{l+1}=s_1$) share a edge in $\UR_{\vec{H}'}$.
	
	For each vertex $s_i \in {\HubSet}'(\vec{H}')$, we define $\Unique(s_j)$ as the subset of vertices of $\vec{H}'$ that are reachable by $s_i$ but are not reachable by any vertex in ${\HubSet}'(\vec{H}')\setminus s_i$. Note that $\Unique(s_i)$ is not empty as we will have $s_i \in \Unique(s_i)$. Also for $s_i \neq s_j$ there can not be any edge in $H$ connecting any vertex in $\Unique(s_i)$ to any vertex in $\Unique(s_j)$ or one of the vertices would be reachable by at least $2$ vertices in ${\HubSet}'(\vec{H}')$.
	
	Similarly, for each pair of vertices $s_i, s_j \in {\HubSet}'(\vec{H}')$, let $\Shared(s_i, s_j)$ be the subset of vertices of $\vec{H}'$ that is reachable by both $s_i$ and $s_j$ but not by any other vertex in ${\HubSet}'(\vec{H}')\setminus\{s_i,s_j\}$. Note that for $i \in [1,l]$ we will have that $\Shared(s_i, s_{i+1})$ is not empty, as they share an edge in $\UR_{\vec{H}'}$. Again for $i\neq j$ there can not be any edge in $H$ connecting any vertex in $\Shared(s_i, s_{i+1})$ with any vertex in $\Shared(s_j, s_{j+1})$ or one of the vertices would be reachable by at least $3$ vertices in ${\HubSet}'(\vec{H}')$.
	
	In \Fig{shared_and_unique} we show an example of the definitions of $\Unique$ and $\Shared$.
	\begin{figure}
		\centering
		\includegraphics[width=\textwidth*2/3]{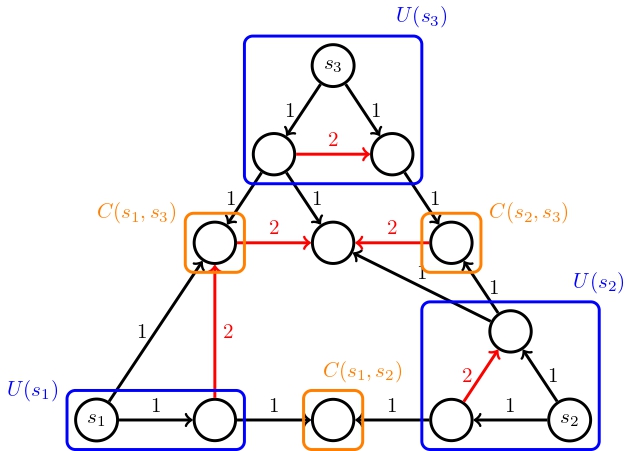}%
		\caption{An example of a graph with three sources that will form the hubset. The $\Unique$ and $\Shared$ regions for each source or pair of sources are highlighted. This graph is $2$-fraternal extension where the Unique Reachability Graph will have a cycle, we can see how the original graph (if we only consider edges with weight $1$) contains an induced cycle of length at least $9$.}
		\label{fig:shared_and_unique}
	\end{figure}
	
	Let $\ReducedV = \bigcup_{i=1}^l \Unique(s_i) \cup \bigcup_{i=1}^l \Shared(s_i, s_{i+1})$. We claim that there is an induced cycle of length $(t+1)\cdot l$ in the induced subgraph of $H$ by $\ReducedV$.
		
	We define $\Area(s_i) = \Shared(s_{i-1}, s_i) \cup \Unique(s_i) \cup \Shared(s_{i}, s_{i+1})$, that is, the portion of $\ReducedV$ reachable by $s_i$.
	
	Now we can show the following claim:
	
	\begin{claim} \label{clm:path_between_shared}
		For all $i\in [1,l]$, for all $u \in \Shared(s_{i-1}, s_i)$ (with $s_0 = s_l$) and for all $w \in \Shared(s_{i}, s_{i+1})$ (with $s_{l+1} = s_1$), exists an induced path in $H$ of length at least $t+1$ connecting $u$ and $w$ that only uses vertices in $\Area(s_i) $.
	\end{claim}
	\begin{proof}
		Let $P$ be the shortest path in $H$ from $u$ to $w$ that only uses vertices in $\Area(s_i) $. Because this is a shortest path it will also be an induced path. Such a path will always exists as there is a path from the vertex $s_i$ to all the vertices in $\Area(s_i)$ that only uses such vertices. We prove that it will have length at least $t+1$:
		
		Assume the opposite, then there is a path from $u$ to $w$ in $P$ that only uses vertices in $\Area(s_i)$ with at most $t$ edges. However by \Clm{reachability-in-paths} we would have that in $\vec{H}'$ there is a path $P' \subseteq P$ where $P'$ either:
		\begin{itemize}
			\item Case 1: Is a direct path from $u$ to $w$, which would mean that $w$ is reachable by three vertices of the hubset and hence not in $\Shared(s_{i}, s_{i+1})$.
			
			\item Case 2: Is a direct path from $w$ to $u$, which would mean that $u$ is reachable by three vertices of the hubset and hence not in $\Shared(s_{i-1}, s_{i})$.
			
			\item Case 3: Has a vertex $s$ that is reachable from both $u$ and $w$, but in that case $s$ would be reachable from three vertices of the hubset and it would not be in $\ReducedV$.
		\end{itemize} 
		
		Hence we reach a contradiction in all the cases and therefore the shortest path must have length at least $t+1$, concluding the proof of this claim.
		
	\end{proof}

	We can now select one vertex $u_i \in \Shared(s_i,s_{i+1})$ for each $i\in [1,l]$ such that the total length for all $i \in [1,l]$ of the paths connecting $u_i$ and $u_{i-1}$ using the vertices in $\Area(s_i)$ is minimized. We have a total of $l$ paths.
	
	If there is a common vertex between two of the paths (ignoring the ends of the paths) then setting that vertex as the end vertex would yield a shorter path, hence we have that the $l$ paths only intersect in their ends, and thus, combining them we obtain a cycle of length $l \cdot (t+1)$ (as each of the individual paths have length $t+1$ by the previous claim).
	
	Only left to show is that this cycle is actually an induced cycle: The paths forming the cycle are all induced, so suffices to show that there are no edges connecting two different paths. There are no edges connecting two different $\Unique(s_i)$ or two different $\Shared(s_i,s_{i+1})$, hence the only possibility would be to have two vertices in the same $\Shared(s_i,s_{i+1})$ of two different paths connected by an edge but if that is the case, replacing the center of that $\Shared$ region by any of the two vertices would reach a shorter total length, which is not possible. Hence we have that the cycle will be an induced cycle.
\end{proof}

\subsection{Rest of the proof}

	In this section we complete the proof of \Lem{main_licl}, the proof is very similar to the proof for Lemma 4.4 in \cite{Bera2021}. For completeness we will include the whole proof with the convenient modifications. We will start by defining a partial \hubtree\, which is a generalization of the partial \dagtree\ introduced in \cite{Bera2021}:
	
	\begin{definition} [partial \hubtree\ ] \label{def:partial_dag}
		Let $\vec{H}$ be a directed graph with hubset set $\HubSet = \HubSet(\vec{H})$. For a subset $\HubSet_p \subseteq \HubSet$, a partial \hubtree\ of $\vec{H}$ with respect to $\HubSet_p$ is a tree $\mathcal{T}=(\mathcal{B},\mathcal{E})$ with the following three properties.
		\begin{enumerate}
			\item Each node $B \in \mathcal{B}$ is a subset of $\HubSet_p$: $B \subseteq S_p$.
			\item The union of the nodes in $\mathcal{T}$ is the entire set $\HubSet_p$: $\bigcup_{B\in \mathcal{B}}B=S_p$.
			\item For all $B,B_1,B_2 \in \mathcal{B}$, if $B$ is on the unique path between $B_1$ and $B_2$ in $\mathcal{T}$, then we have $\Reachable_{\vec{H}}(B_1)\cap \Reachable_{\vec{H}}(B_2) \subseteq \Reachable_{\vec{H}}(B)$.
		\end{enumerate}
	\end{definition}
	
	In the case that $\HubSet_p = \HubSet$ this definition corresponds exactly with the \hubtree.
	
	Now we bring a few more definitions from \cite{Bera2021}, again generalized to directed graphs:

	\begin{definition}[Intersection-cover and $\HubSet_p$-cover] \label{def:inter_cover}
		Let $\vec{H}$ be a directed graph with hubset $\HubSet$ = $\HubSet(\vec{H})$. Let $s_1$ and $s_2$ be a pair of vertices in $\HubSet$. We call a vertex $s\in \HubSet$ an intersection-cover of $s_1$ and $s_2$ if $\Reachable(s_1) \cap \Reachable(s_2) \subseteq \Reachable(s)$. Assume $\HubSet_p \subseteq \HubSet$ is a subset of the hubset $\HubSet$. We call a vertex $s \in \HubSet$, a $\HubSet_p$-cover of $s_1 \in \HubSet$ if for each vertex $s_2 \in \HubSet_p$, $s$ is an intersection-cover for $s_1$ and $s_2$.
	\end{definition}

	\begin{definition} [Good-pair] \label{def:good_pair}
		Let $\vec{H}$ be a directed graph with hubset $\HubSet$ = $\HubSet(\vec{H})$. Let $x \in \HubSet$ be a vertex of the hubset and $\mathcal{T}_{\HubSet_p}$ be a partial \hubtree\ of width one for $\HubSet_p \subset \HubSet$ where $x \notin \HubSet_p$. We call the pair $(x, \mathcal{T}_{\HubSet_p})$ a good-pair if there exists a leaf node $l \in \mathcal{T}_{\HubSet_p}$ connected to the node $d \in \mathcal{T}_{\HubSet_p}$ such that $d$ is an intersection-cover for $x$ and $l$.
	\end{definition}

	We also restate the following Lemma, presented in \cite{Bera2021} as Lemma $4.8$, but extending it to non-acyclical directed graphs:
	
	\begin{lemma} \label{lem:adding_source} [Equivalent to Lemma $4.8$ of \cite{Bera2021}]
		Let $\vec{H}$ be a directed graph with hubset $\HubSet = \HubSet(\vec{H})$ and let $\HubSet_p \subset \HubSet$ be a subset of the hubset. Assume $\cT$ is a partial \hubtree\ for $\HubSet_p$ with $\htw(\cT)=1$. Consider a vertex $s \in \HubSet$ such that $s \not\in \HubSet_p$. If $d$ is a $\HubSet_p$-cover of $s$, then connecting $s$ to $d$ in $\cT$ as a leaf results in a tree $\cT'$ that is a partial \hubtree\ for $\HubSet_p \cup \{s\}$. Furthermore, $\htw(\cT')=1$
	\end{lemma}
	\begin{proof}
		Because we had $\htw(\cT)=1$ and we are just adding a leaf with a single vertex to $\cT$, we will have that $\htw(\cT')=1$. Therefore, suffices to show that $\cT'$ is a valid partial \hubtree\ for $\HubSet_p \cup \{s\}$. We prove by contradiction: Assume it is not, then there must exist three nodes $s_1,s_2,s_3 \in \cT'$, with $s_3$ being in the path between $s_1$ and $s_2$, such that $\Reachable(s_1)\cap \Reachable(s_2) \not\subseteq \Reachable(s_3)$. If $s \not\in \{s_1,s_2,s_3\}$, then this is not possible as all the nodes were already in $\cT$ and it was a valid \hubtree{}. Hence $s$ must be one of the three vertices, it can not be $s_2$ as $s$ is a leaf, we can assume without loss of generalization that $s=s_3$. Now $s_2$ lies on the unique path between $s_1$ and $s$. 
		
		If $s_2=d$, then because $d$ is a $\HubSet_p$-cover of $s$ we will have that $\Reachable(s_1)\cap \Reachable(s) \subseteq \Reachable(d)$, reaching a contradiction. Otherwise, $s_2\neq d$ but $s_2$ must lie in the path between $s_1$ and $d$, and hence $\Reachable(s_1)\cap \Reachable(d) \subseteq \Reachable(s_2)$ and we also had that $\Reachable(s_1)\cap \Reachable(s) \subseteq \Reachable(d)$ hence $\Reachable(s_1)\cap \Reachable(s) \subseteq \Reachable(s_2)$. Reaching a contradiction.
	\end{proof}

	Now we can prove the main lemma of this section, again following closely the proof of Lemma $4.4$ of \cite{Bera2021}. We restate the lemma:
	
	\lemmalicl*
	\begin{proof}
		Let $\vec{H}'$ be a $t$-fraternal extension of $H$ and $\HubSet=\HubSet(\vec{H}')$ be the hubset of $\vec{H}'$. Let $\HubSet_p \subseteq \HubSet$ denote a subset of $\HubSet$. We prove by induction on the size of $\HubSet_p$, that there exists a partial \hubtree\ of width $1$ for each $\HubSet_p \subseteq \HubSet$. If $\HubSet_p = \HubSet$ then we have that there is a \hubtree\ for $\vec{H}'$ of width $\htw = 1$.
		
		The base cases for $|\HubSet_p|=1$ and $|\HubSet_p|=2$ are both trivial: for $|\HubSet_p|=1$ we can put the only vertex of $\HubSet_p$ in its own bag and it will be a valid partial \hubtree. For $|\HubSet_p|=2$ we can put both vertices in separate bags and connect them by an edge, obtaining again a valid partial \hubtree.
		
		For the inductive step we assume that it is possible to build a partial \hubtree\ with \hubtreewidth\ one for any subset $\HubSet_p \subset \HubSet$ where $|\HubSet_p|\leq r$, and $1 \leq |r| < |\HubSet|$. We show how to construct a partial \hubtree\ with $\htw=1$ for any subset of $\HubSet$ of size $r+1$:
		
		Let $\HubSet_{r+1} \subseteq \HubSet$ be any subset of size $r+1$. Let $x \in \HubSet_{r+1}$ be an arbitrary vertex of $\HubSet_{r+1}$. By the induction hypothesis we can construct a partial \hubtree\ with $\htw=1$ for $\HubSet_{-x}=\HubSet_{r+1}\setminus\{x\}$. We denote such \hubtree\ with $\cT_{-x}$. We can then show that $(x, \cT_{-x})$ forms a good-pair, this is given by the following claim which is equivalent to Claim $4.10$ of \cite{Bera2021}:
		
		\begin{claim}
			There exists a vertex $x \in \HubSet_{r+1}$ and a width one partial \hubtree\ $\cT_{-x}$ for $\HubSet_{-x} = \HubSet_{r+1}\setminus \{x\}$ such that $(x, \cT_{-x})$ is a good-pair.
		\end{claim}
		\begin{proof}
			We prove by contradiction. Assume that the claim is false, consider the unique reachability graph on the vertex set $\HubSet_{r+1}$, $\UR_{\HubSet_r+1}$. Let $x \in \HubSet_{r+1}$ be an arbitrary vertex from $\HubSet_{r+1}$. By the assumption we have that $(x, \cT_{-x})$ is not a good-pair. Hence, for each leaf node $l \in \cT_{-x}$ connected to the vertex $d$ we get that $d$ is not an intersection-cover for $x$ and $l$. Hence there exist a vertex $v$ that is reachable by $x$ and $l$ but not $d$. But also, because $d$ is the only vertex connected to $l$ in $\cT_{-x}$ we have that $d$ is a $\HubSet_{-x}$-cover for $l$ and the only vertex that can reach $v$ in $\HubSet_{-x}$ is $l$. Thus, the edge $\{x,l\}$ will be in $\UR_{\HubSet_r+1}$. 
			
			Because $\cT_{-x}$ has at least two leaves we will have that the degree of $x$ in $\UR_{\HubSet_{r+1}}$ must be at least $2$. This is true for every vertex in $x \in \HubSet_{r+1}$. This implies that there is a cycle in $\UR_{S_{r+1}}$ of length at least $3$, using \Lem{licl_to_acyclic} this means that $\LICL(H) \geq 3(t+1)$, but we had that $\LICL(H) < 3(t+1)$, hence reaching a contradiction.
		\end{proof}
		
		Now, we show that if we have a good-pair $(x, \cT_{-x})$ we can construct a \hubtree\ of \hubtree\ one, again \cite{Bera2021} proved a more restrictive statement that we will generalize:
		
		\begin{claim}[Equivalent to Claim $4.9$ from \cite{Bera2021}]
			Let $x \in \HubSet_{r+1}$ and $\cT_{-x}$ be a width one partial \hubtree\ for $\HubSet_{-x} = \HubSet_{r+1} \setminus \{x\}$ such that $(x,\cT_{-x})$ is a good-pair. Then, there exists a partial \hubtree\ $\cT$ for $\HubSet_{r+1}$ with $\htw(\cT)=1$.
		\end{claim} 
		\begin{proof}
			For $(x, \cT_{-x})$ to form a good-pair we must have that there is a leaf $l$ in $\cT_{-x}$ connected to a node $d \in \cT_{-x}$ such that $d$ is an intersection-cover for $x$ and $l$. From the assumption of the inductive step we can construct a \hubtree\ of width one for $\HubSet_{r+1}\setminus l$ and connect $l$ as a leaf to the node $d$. Let $\cT$ be the resultant tree. We can show that $\cT$ is a valid \hubtree\ of $\HubSet_{r+1}$ of width $1$.
			
			We had that $d$ is intersection-cover of $l$ and $x$, also because $l$ only connects to $d$ in $\cT_{-x}$ we have that $d$ is a $\HubSet_{-x}$-cover of $l$. Hence $d$ is a $\HubSet_{r+1}$-cover of $l$.
			
			By \Lem{adding_source} we have that $\cT$ is a valid partial \hubtree\ of $\HubSet_{r+1}$ with \hubtreewidth\ one.
		\end{proof}
		
		Hence, combining both claims we get that we can construct a \hubtree\ with \hubtreewidth\ of 1 for $\HubSet_{r+1}$. This proves the induction and subsequently the lemma.
	\end{proof}

\section{Generalizing Bressan's algorithm} \label{sec:bressan}

In this section we prove \Lem{main_bressan}. This will complete the proof of the upper bound of our Main Theorem, as shown in \Sec{upper}. We will show how to adapt Bressan's Algorithm to compute the homomorphisms of the fraternal extensions. This requires working with non-acyclical graphs using the hubset and the \hubtree\ instead of the \dagtree, and using graphs that are weighted and labeled. We start by restating the main lemma of this section:

\lemmaBressan*

Given directed graphs $\vec{H}$ and $\vec{G}$. Lemma $4$ in \cite{Bressan2021} shows a way of computing homomorphisms for the subgraphs induced by $\Reachable_{\vec{H}}(s)$ for every source $s \in S(\vec{H})$. We can generalize this result to directed weighted and labeled graphs:

\begin{lemma} \label{lem:hom_Hs}
	Let $\vec{H}$ be a directed weighted and labeled graphs with $k$ vertices and hubset $\HubSet = \HubSet(\vec{H})$. Let $\vec{G}$ be a directed weighted graph with max outdegree $d = \Delta^+(\vec{G})$. For any vertex $s \in \HubSet$, the set of homomorphisms from $\vec{H}(s)$ to $\vec{G}$ has size $O(d^{k-1}n)$ and can be enumerated in time $O(k^2d^{k-1}n)$.
\end{lemma}
\begin{proof}
	Let $\vec{T}$ be a directed spanning tree of $\vec{H}(s)$ rooted at $s$. Let $O$ be any arbitrary ordering of the vertices of $\vec{H}(s)$ such that all the edges of $T$ are not inverted. For every vertex $u \in \vec{H}(s)$ following that ordering, we can enumerate all the possible candidates of $\vec{G}$ for the mapping $\phi(u)$. The first vertex $s$ will have $n$ candidates, as it can be assigned to every vertex in $\vec{G}$. However for all the other vertices, because they have at least one incoming edge from a vertex already assigned, we just need to look at the out-neighbors of the corresponding mapped vertex in $\vec{G}$, there will be then at most $d$ candidates, as that is the maximum outdegree in $\vec{G}$. Hence the total number of possible homomorphisms is bounded by $O(nd^{k-1})$. We can list all these candidate homomorphisms in a similar amount of time. Only left is to verify if each candidate homomorphisms $\phi$ is valid:
	\begin{itemize}
		\item For each vertex $u \in V_{\vec{H}(s)}$, verify that they are mapped to a vertex with the same label $L_{\vec{H}}(u) = L_{\vec{G}}(\phi(u))$.
		\item For each edge $(u,v) \in E_{\vec{H}(s)}$, verify that $(\phi(u),\phi(v))\in E_{\vec{G}}$ and $W_{\vec{H}}((u,v)) \geq W_{\vec{G}}((u,v))$.
	\end{itemize}
	This can be done in $O(k^2)$ time as we will have at most $k$ vertices and $k^2$ edges, and every check can be done in constant time. Hence the total time required will be $O(k^2d^{k-1}n)$.
\end{proof}	

Note that in the case that the graph $\vec{G}$ has bounded outdegree and $\vec{H}$ is constant sized we will be able to compute $\Hom{\vec{G}}{\vec{H}(s)}$ in $O(n)$ time for all the vertices $s\in \HubSet(\vec{H})$.

Given a \hubtree\ $\cT$ of $\vec{H}$, we will use $\down(B)$ to denote the down-closure of $B$ in $\cT$, that is, the union of all the bags $B \in \cB$ that are descendants of $B$. We will then use $\vec{H}(\down(B))$ to refer to the union of all the graphs $\vec{H}(B)$ for $B \in \down(B)$.

If $\vec{H}$ has $\htw(H)=1$ we can used a modification of the algorithm presented by Bressan in \cite{Bressan2021} to compute $\Hom{\vec{G}}{\vec{H}}$ in linear time. Given a \hubtree\ decomposition $\cT$ of $\vec{H}$, this algorithm uses dynamic programming to compute $\Hom{\vec{G}}{\vec{H}(\down(s))}$ for any vertex $s \in \HubSet$ aggregating the values of $\Hom{\vec{G}}{\vec{H}(\down(s'))}$ of all the descendants $s'$ of $s$ and $\Hom{\vec{G}}{\vec{H}(s)}$.

Given a homomorphism $\phi$ we say that $\phi'$ respects $\phi$ if for every value $u$  that $\phi$ takes $\phi(u)=\phi'(u)$. Given a homomorphism $\phi$ that maps the vertices in the set $V$, we call the restriction of $\phi$ to $V'\subseteq V$ to the map $\phi'$ that maps the vertices of $V'$ with $\phi'(v)=\phi(v) \forall v\in V'$. Additionally we denote with $\ext(\vec{H},\vec{G},\phi)$ to the number of homomorphisms $\phi'$ from $\vec{H}$ to $\vec{G}$ that respects $\phi$. We can show the following lemma which is a generalization of Lemma $3$ in \cite{Bressan2021}:

\begin{lemma} \label{lem:bressan_extension}
	Let $\cT$ be a \hubtree\ of a graph $\vec{H}$ and let $B_1,\ldots,B_l$ be the children of $B$ in $\cT$. Fix $\phi_B: \vec{H}(B) \to \vec{G}$. Let $\Phi(\phi_B)=\{\phi:\vec{H}(\down(B)) \to \vec{G}| \phi \text{ respects }\phi_B)\}$, and for $i=1,\ldots,l$ let $\Phi_i(\phi_B)=\{\phi:\vec{H}(\down(B_i))\to \vec{G}| \phi \text{ respects }\phi_B)\}$. Then there exists a bijection between $\Phi(\phi_B)$ and $\Phi_1(\phi_B)\times \ldots\times \Phi_l(\phi_B)$, and therefore:
	\[
		\ext(\vec{H}(\down(B)),\vec{G},\phi_B) = \prod_{i=1}^{l} \ext(\vec{H}(\down(B_i)),\vec{G},\phi_B)
	\]
\end{lemma}
\begin{proof}
	The proof is similar to the proof of Lemma $3$ in \cite{Bressan2021}:
	
	First, we show that there exists an injection from $\Phi(\phi_B)$ to $\Phi_1(\phi_B)\times \ldots\times \Phi_l(\phi_B)$: Fix any $\phi \in \Phi(\phi_B)$, and let $\phi_i$ be the restriction of $\phi$ to $\vec{H}(\down(B_i))$, because $\phi$ respects $\phi_B$ so will $\phi_i$, hence $\phi_i \in \Phi_i(\phi_B)$, therefore the tuple $(\phi_1,\ldots,phi_l) \in \Phi_1(\phi_B)\times \ldots\times \Phi_l(\phi_B)$.
	
	Now we show the opposite, that there exists an injection from $\Phi_1(\phi_B)\times \ldots\times \Phi_l(\phi_B)$ to $\Phi(\phi_B)$: Fix any tuple $(\phi_1,\ldots,\phi_l) \in \Phi_1(\phi_B)\times \ldots \times \Phi_l(\phi_B)$, note that the different $\phi_i$ of the tuple only intersect in $\vec{H}(B)$, and they all respect $\phi_B$, hence we can combine $\phi_B, \phi_1,\ldots,\phi_l$ and obtain a homomorphism $\phi$ from $\vec{H}(\down(B_i))$ to $\vec{G}$ such that $\phi$ respects $\phi_B$, hence we will have that $\phi \in  \Phi(\phi_B)$. 
	
\end{proof}

We now adapt Bressan's algorithm \cite{Bressan2021} to our setting, the full algorithm can be seen in \Alg{bressan}. We prove its correctness and runtime in the following Lemma:

\begin{lemma} \label{lem:bressan_alg}
	Let $\vec{H}$ be a labeled weighted and directed graph with a \hubtree\ $\cT$ such that $\htw(\cT)=1$, let $B$ be any node of $\cT$ and let $\vec{H}$ be a labeled weighted and directed graph with bounded outdegree.  \Alg{bressan} returns a dictionary $C_B$ such that for every homomorphism $\phi \in \Phi(\vec{H}(B),\vec{G})$ we have $C_B(\phi)=\ext(\vec{H}(\down(B)),\vec{G},\phi)$, and runs in $O(n\cdot \log{n})$ time.
\end{lemma}
\begin{proof}
	First we prove the correctness of the algorithm. We can see that in the base case, when $B$ is a leaf of $\cT$, $C_B$ will contain $1$ for every $\phi_B \in \Phi(H(B),G)$. 
	
	If $B$ is not a leaf, we assume that the algorithm returns the desired value for every child $B_i$ of $B$. In this case the value of $\Agg_{B_i}$ after the first for loop will be $|\phi \in \Phi(\vec{H}(B_i),\vec{G}):\phi \text{ respects }\phi_r| = \ext(\vec{H}(\down(B_i)),\vec{G},\phi_r)$. Hence we will have that:
	\[
		C_B(\phi) = \prod_{i=1}^l \Agg_{B_i}(\phi_i) =  \prod_{i=1}^l \ext(\vec{H}(\down(B_i)),\vec{G},\phi_i) 
	\]
	\[
	= \prod_{i=1}^l \ext(\vec{H}(\down(B_i)),\vec{G},\phi) = \ext(\vec{H}(\down(B)),\vec{G},\phi)
	\]
	Where the last inequality comes from \Lem{bressan_extension}.
	
	For the runtime, we have that $B$ has at most $O(k)$ children, from \Lem{hom_Hs} we have that every dictionary $C_{B_i}$ will have at most $O(n)$ keys, and we can enumerate them in $O(n)$ time. We will need $O(n\log{n})$ time to access the dictionary, so the total complexity is $O(n\cdot \log{n})$.
\end{proof}

\begin{algorithm}[H]
	\hspace*{\algorithmicindent} \textbf{Input:} \\
	\hspace*{6ex}-Directed weighted labeled graph $\vec{H}$ with \hubtree\ $\cT$\\
	\hspace*{6ex}-Directed weighted labeled graph $\vec{G}$ \\
	\hspace*{6ex}-A node $B \in \cT$ \\
	\hspace*{\algorithmicindent} \textbf{Output:}  \\
	\hspace*{6ex}-Dictionary $C_B$
	\begin{algorithmic}[1]
		\State Let $C_B$ be an empty dictionary with default value $0$.
		\If{$B$ is a leaf}
		\For{every homomorphism $\phi_B: \vec{H}(B) \to \vec{G}$}
		\State $C_B(\phi_B)=1$
		\EndFor
		\Else
		\State let $B_1,\ldots,B_l$ be the children of $B$ in $\cT$
		\For{$i=1,\ldots,l$}
		\State $C_{B_i} = ${Generalized Bressan's Algorithm}$(\vec{H},\vec{G},B_i)$
		\State Let $\Agg_{B_i}$ be an empty dictionary with default value $0$.
		\For{every key $\phi$ in $C_{B_i}$}
		\State let $\phi_r$ be the restriction of $\phi$ to $\Reachable_{\vec{H}}(B)\cap\Reachable_{\vec{H}}(\down(B_i))$
		\State $\Agg_{B_i}(\phi_r) += C_{B_i}(\phi)$
		\EndFor
		\EndFor
		\For{every homomorphism $\phi:\vec{H}(B)\to \vec{G}$}
		\State Let $\phi_i$ be the restriction of $\phi$ to $\Reachable_{\vec{H}}(B)\cap\Reachable_{\vec{H}}(\down(B_i))$, for $i=1,\ldots,l$.
		\State $C_B(\phi) = \prod_{i=1}^l \Agg_{B_i}(\phi_i)$
		\EndFor
		\EndIf
		\State return $C_B$
	\end{algorithmic}
	\caption{Generalized Bressan's Algorithm: \Homomorphisms($\vec{H},\vec{G},B$) } \label{alg:bressan}
\end{algorithm}

We can finally prove the main lemma of this section:

\begin{proof} [Proof of Lemma \ref{lem:main_bressan}]
	We can compute a \hubtree\ $\cT$ for $\vec{H}$ in $f(k)$ time for some function $f$ and then run \Alg{bressan} in the root $s$ of $\cT$ to obtain $C_s$. From \Lem{bressan_alg} we have that this takes $O(n\cdot \log{n})$. We can then sum all the values of $C_s$ to obtain $\Hom{\vec{H}}{\vec{G}}$, this takes additional $O(n)$ time.
	
	We prove the correctness of this approach: Because $s$ is the root of $\cT$ we will have that $\vec{H}(\down(s))=\vec{H}$, hence $C_s(\phi) = \ext(\vec{H},\vec{G},\phi)$ for all $\phi \in \Phi(\vec{H}(s),\vec{G})$. Summing over all $\phi$ we have that:
	\[
	\sum_{\phi \in \Phi(\vec{H}(s),\vec{G})} C_s(\phi) = 
	\sum_{\phi \in \Phi(\vec{H}(s),\vec{G})}\ext(\vec{H},\vec{G},\phi) = 
	\Hom{\vec{G}}{\vec{H}}
	\] 
\end{proof}

\section{Lower Bound} \label{sec:lower}

In this section we prove the lower bound of the main theorem, given by the following theorem:

\begin{theorem} \label{thm:lower}
	For all $t>0 \in \mathbb{N}$, let $G$ be an input graph with $n$ vertices, $m$ edges and bounded $\grad_{(t-1)/2}(G)$ and let $H$ be a pattern graph on $k$ vertices with $LICL(H) \geq 3(t+1)$. Assuming  the \TRICONJ, there exists an absolute constant $\gamma > 0$ such that there is no (expected) $o(m^{1+\gamma})$ algorithm for the $\Hom{G}{H}$ problem.
\end{theorem}

For $t=1$ the former theorem is proved to be true \cite{Bera2021} as bounded $\grad_0$ is equivalent to bounded degeneracy. We will prove that the theorem holds for all the $t>1$.

In order to do so we first show that counting a pattern in a bounded grad class is as hard as counting any subgraph of such pattern, this was showed to be true by \cite{Bera2022} in the case of bounded degeneracy graphs. The proof uses some techniques introduced by \cite{curticapean2017homomorphisms}. We will extend such proof for all bounded grad classes.

Then we will show a simple reduction inspired by \cite{bera2020linear} that allows to relate counting triangles in a general graph (a problem which can not be done in linear time if the \TRICONJ{} is true) with counting non-induced cycles in bounded grad classes. We then finalize the proof by extending te result to homomorphism counts of cycles.

\subsection{Reducing to Cycle Homomorphisms}

In this subsection we show that computing homomorphisms of a pattern is as easy as computing all homomorphisms of all the induced subgraphs of that pattern. The proof follows closely the proof from Lemma $1.7$ in \cite{Bera2022}, but generalizing to graphs with bounded $\grad_k$.

\begin{lemma} \label{lem:hom_subgraphs}
	Let $\cG_i$ be the class of bounded $\grad_i$ graphs, for some $i$. Let $H$ be a pattern graph, if computing $\Hom{G}{H}$ for any $G \in \cG_i$ is easy, then so is computing $\Hom{G}{H'}$ for every induced subgraph $H'$ of $H$.
\end{lemma}

We prove this lemma by proving the more general lemma that follows, which is a generalization of Lemma $4.1$ from \cite{Bera2022}:

\begin{lemma} \label{lem:generalization_4_1}
	For every graph $H$ there is $k=k(H)$ such that the following hols. For every graph $G$ there are graphs $G_1,\ldots,G_k$, computable in time $O(|V_G|+|E_G|)$, such that $|V_{G_i}|=O(|V_G|)$ and $|E_{G_i}|=O(|E_G|)$ for every $i=1,\ldots,k$ such that knowing $\Hom{G_1}{H},\ldots,\Hom{G_k}{H}$ allows one to compute $\Hom{G}{H'}$ for all induced subgraphs $H'$ of $H$ in constant time. Furthermore, if $G$ has bounded $\grad_i$, then so do $G_1,\ldots,G_k$.
\end{lemma}

In order to prove this lemma, we first need to introduce another additional lemma, which again is a generalization of a lemma from \cite{Bera2022}, in this case Lemma $4.2$:

\begin{lemma} \label{lem:generalization_4_2}
	Let $H_1,\ldots,H_k$ be pairwise non-isomorphic graphs and let $c_1,\ldots,c_k$ be non-zero constants. For every graph $G$ there are graphs $G_1,\ldots,G_k$, computable in time $O(|V_G|+|E_G|)$ such that $|V_{G_i}|=O(|V_G|)$ and $|E_{G_i}|=O(|E_G|)$ for every $i=1,\ldots,k$, and such that knowing $b_j:=c_1 \cdot \Hom{G_j}{H_1}+\ldots+c_k \cdot\Hom{G_j}{H_k}$ for every $j=1,\ldots,k$ allows one to compute $\Hom{G}{H_1},\ldots,\Hom{G}{H_k}$ in constant time. Furthermore, if $G$ has bounded $\grad_i$, then so do $G_1,\ldots,G_k$.
\end{lemma}
\begin{proof}
	We start the proof by stating the following lemma from \cite{erdHos1979strong} and \cite{lovasz2012large}, which was stated in \cite{Bera2022} as Lemma $A.2$:
	
	\begin{lemma} \label{lem:lovasz}
		Let $H_1,\ldots,H_k$ be pairwise non-isomorphic graphs, and let $c_1,\ldots,c_k \neq 0$ be non-zero constants. Then there exist graphs $F_1,\ldots,F_k$ such that the $k\times k$ matrix $M_{i,j}=c_j \cdot \Hom{F_i}{H_j}, 1\leq i,j\leq k$, is invertible.
	\end{lemma}
	
	Now, let $G_i = F_i \times G$, we first show that if $G$ has bounded $\grad_i$, then so do $G_1,\ldots,G_k$: The proof is very similar to \Lem{labeled_grad}. Note that the size of $F_i$ does not depends on the input graph $G$, only on the graphs $H_1,\ldots,H_k$. Because we assume such graphs to be constant-sized so will be $F_1,\ldots,F_k$. Let $f$ be the number of vertices in $F_i$, we have:
	\[
		G_i = F_i \times G \subseteq G \times K_f = G \bullet \bar{K}_f \subseteq G \bullet K_f
	\]
	Where the second equality comes from \Fact{products}. Then using \Prop{grad} we will have that if $\grad_i(G)$ is bounded so will $\grad_i(G \bullet K_f)$, and therefore by the previous equation so will $\grad_i(G_i)$.
	
	Now, let $b_i = \sum_{j=1}^{k} c_1 \cdot \Hom{G_j}{H_1}+\ldots+c_k \cdot \Hom{G_j}{H_k}$, we can rewrite it as:
	\[
		b_i = \sum_{j=1}^k c_j\cdot\Hom{F_i\times G}{H_j} = \sum_{j=1}^k c_j\cdot\Hom{F_i}{H_j} \cdot \Hom{G}{H_j} = \sum_{j=1}^k M_{i,j} \cdot \Hom{G}{H_j} 
	\]
	Hence for $1\leq i \leq k$ we obtain a system of linear equations with $\Hom{G}{H_1},\ldots,\Hom{G}{H_k}$ as variables and $M$ as the matrix of the system. By \Lem{lovasz} $M$ is invertible, then given $b_1,\ldots,b_k$ we can compute $\Hom{G}{H_1},\ldots,\Hom{G}{H_k}$ in constant time.
\end{proof}

We can now complete the proof by proving \Lem{generalization_4_1}:

\begin{proof}[Proof of Lemma \ref{lem:generalization_4_1}]
	The proof of this lemma comes directly from Lemma $4.1$ in \cite{Bera2022}. The only difference is showing that the graphs $G_1,\ldots,G_k$ have bounded grad, instead of bounded degeneracy, when $G$ does. We can see that this comes directly from \Lem{generalization_4_2}, which generalizes Lemma $4.2$ in \cite{Bera2022}.
\end{proof}

\subsection{From counting triangles to counting cycles}

In this subsection we prove a hardness result for non-induces copies in the case that the pattern is the cycle graph. We use a reduction very similar to the one found in \cite{bera2020linear}: we can take any graph $G$ and replace every edge by some combination of paths. The resultant graphs will actually have bounded grad for certain depth, depending on the length of the path.

We will prove the following:

\begin{lemma} \label{lem:sub_cycles}
	For all $t>1 \in \mathbb{N}$, let $G$ be any input graph with $n$ vertices, $m$ edges and bounded $\grad_{(t-1)/2}(G)$ and let $H = \cC_k$ be the cycle graph on $k$ vertices for $ k \in \{3(t+1), 3(t+1)+1\}$. Assuming  the \TRICONJ, there exists an absolute constant $\gamma > 0$ such that there is no (expected) $o(m^{1+\gamma})$ algorithm for the $\Sub{G}{H}$ problem.
\end{lemma}

\begin{proof}
	Fix any $t > 1$. We will first show a reduction for cycles of length $k=3(t+1)$.
	
	Let $G$ be a graph. We define the graph $\Gt$ by replacing every edge in $G$ by a path of $t+1$ edges, formally:
	\begin{definition} [$\Gt$]
		Let $G = (V,E)$ be an arbitrary input graph. We define the reduced graph $\Gt = (V_t,E_t)$ as follows:
		\begin{itemize}
			\item For each vertex $v \in V$ we create a vertex $v \in V^o$.
			\item For each edge $e=(u,w) \in E$ we create $t$ extra vertices $v_{e,1},\ldots,v_{e,t}$ in $V^*$.
			\item We define $V_t = V^o \cup V^*$.
			\item We create the edge set $E_t$ by adding the edges $(u,v_{e,1}),(v_{e,1},v_{e_2}),\ldots,(v_{e,t-1},v_{e,t}),(v_{e,t},w)$ for every edge $(u,v)\in E$.
		\end{itemize}
	\end{definition}
	
	\Fig{reduction} shows how the reduction replaces an edge by the path in $\Gt$. We can show that there exists a relation between the number of triangles in $G$ and the number of $\cC_k$ cycles in $\Gt$:
	
	\begin{claim} \label{clm:triangles_to_cycles_1}
		Let $t>1$ and $G$ any graph, set $k=3(t+1)$, there is a triangle in $G$ if and only if there is a $\cC_k$ cycle in $\Gt$.
	\end{claim}
	\begin{proof}
		Consider any triangle $v_1,v_2,v_3$ in $G$, in the graph $\Gt$ each pair of those vertices will be separated by a path of length $t+1$, hence combining those paths we obtain a cycle of length $3(t+1) = k$ in $\Gt$.
		
		Conversely, let $\cC = \cC_k$ be a $k$-cycle in $\Gt$, we can show that $\cC$ must contain exactly $3$ vertices in $V^o$: If it contained $1$ or $2$ then $\cC$ would not be able to be a cycle, while if it contains $4$ or more then it will form a cycle of at least $4(t+1)$ vertices, which is greater than $k$. Take the three vertices in $V^o$, they must be connected to each other by a path of $t+1$ edges in $\Gt$ and hence by edges in $G$, therefore, they will form a triangle.
	\end{proof}
	
	Similarly we define a reduction for the cycles of length $k=3(t+1)+1$. In this case we will replace every edge of $G$ by two different paths:
	
	\begin{definition} [$\Gtt$]
		Let $G = (V,E)$ be an arbitrary input graph. We define the reduced graph $\Gtt = (V_{t'},E_{t'})$ as follows:
		\begin{itemize}
			\item For each vertex $v \in V$ we create a vertex $v \in V^o$.
			\item For each edge $e=(u,w) \in E$ we create $t$ extra vertices $v_{e,1},\ldots,v_{e,t}$ in $V^*$ and another additional $t+1$ vertices $v'_{e,1},\ldots,v'_{e,t+1}$ in $V^*$
			\item We define $V_t = V^o \cup V^*$.
			\item We create the edge set $E_t$ by adding the edges $(u,v_{e,1}),(v_{e,1},v_{e_2}),\ldots,(v_{e,t-1},v_{e,t}),(v_{e,t},w)$ and $(u,v'_{e,1}),(v'_{e,1},v'_{e_2}),\ldots,(v'_{e,t},v'_{e,t+1}),(v'_{e,t+1},w)$ for every edge $(u,v)\in E$.
		\end{itemize}
	\end{definition}
	
	In \Fig{reduction} we show how each edge of $G$ is replaced in $\Gtt$. Again we can show that there is a relation between the number of triangles in $G$ and the number of $\cC_k$ cycles in $\Gtt$:
	
	\begin{figure}
		\centering
		\includegraphics[width=\textwidth*3/4]{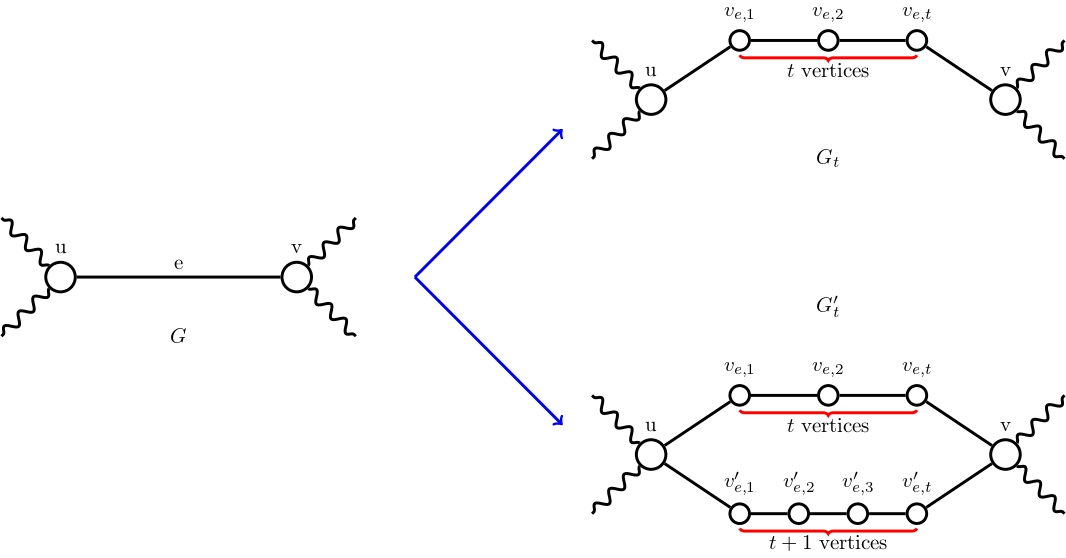}%
		\caption{An example of how an edge $e=(u,v)$ in $G$ is replaced in the reduced graphs $\Gt$ and $\Gtt$. In $\Gt$ we will add $t$ vertices between $u$ and $v$ forming a path. While in $\Gtt$ we will add two path, one with $t$ vertices and another with $t+1$ vertices. This process will be applied to every edge in $G$.}
		\label{fig:reduction}
	\end{figure}

	\begin{claim} \label{clm:triangles_to_cycles_2}
		Let $t>1$ and $G$ any graph, set $k=3(t+1)+1$, there is a triangle in $G$ if and only if there is a $\cC_k$ cycle in $\Gtt$.
	\end{claim}
	\begin{proof}
		Consider any triangle $v_1,v_2,v_3$ in $G$, in the graph $\Gtt$ each pair of those vertices will be separated by a path of length $t+1$ and a path of length $t+2$, hence combining those paths we can obtain three different cycles of length $3(t+1)+1 = k$ in $\Gtt$.
		
		Conversely, let $\cC = \cC_k$ be a $k$-cycle in $\Gtt$, we can show that $\cC$ must contain exactly $3$ vertices in $V^o$: If it contained $1$ or $2$ then $\cC$ could only be a cycle of length $2t+3$ which is strictly less than $3(t+1)+1$ for $t>1$, while if it contains $4$ or more then it will form a cycle of at least $4(t+1)$ vertices, which is greater than $3(t+1)+1=k$. Take the three vertices in $V^o$, they must be connected to each other by either a path of $t+1$ or $t+2$ edges in $\Gtt$ and hence by edges in $G$, therefore, they will form a triangle.
	\end{proof}

	 We also show that both $\Gt$ and $\Gtt$ have bounded $\grad_{(t-1)/2}$:
	
	\begin{claim} \label{clm:G_bounded_grad_1}
		Let $t>1$ and $G$ be an arbitrary graph, $\Gt$ and $\Gtt$ have bounded $\grad_{(t-1)/2}$.
	\end{claim}

	\begin{proof}
		Let $G'=(V',E')$ be any shallow topological minor of $\Gt$ or $\Gtt$ at depth $(t-1)/2$. That is, a graph where the vertices are a subset of the vertices of $\Gt$ or $\Gtt$ and the edges correspond to disjoint paths in $\Gt$ or $\Gtt$ of length at most $t$. 
		
		The vertices of $G'$ can either be part of ${V'}^*=V^*\cap V'$ or ${V'}^o=V^o\cap V'$. The degree of the vertices in ${V'}^*$ can not be greater than $2$ as the original degree of such vertices in either $\Gt$ or $\Gtt$ were $2$. Additionally, any edge in $G'$ can not have both of its endpoints in ${V'}^o$, as the minimum distance between such vertices is $t+1$. Therefore, every edge will have at least one end in ${V'}^*$, because the degree of such vertices is at most $2$, we will have at most $2|{V'}^*|$ edges.
	
		We can then bound the average edge density of $G'$:
		
		\begin{align*}
			&\frac{|E'|}{|V'|} 
			\leq \frac{ 2 |{V'}^*|}{|{V'}^*|+|{V'}^o|} 
			\leq \frac{ 2 |{V'}^*|}{|{V'}^*|} 
			= 2
		\end{align*}
		
		Hence every topological minor of $\Gt$ or $\Gtt$ at depth $(t-1)/2$ has bounded average edge density and $\topgrad_{(t-1)/2}(\Gt)$ and $\topgrad_{(t-1)/2}(\Gtt)$ are bounded. Which by \Fact{grad_topgrad} implies that $\grad_{(t-1)/2}(\Gt)$ and $\grad_{(t-1)/2}(\Gtt)$ are also bounded.
	\end{proof}
	
	Now, assume that we have an algorithm that can count the number of cycles of size $k=3(t+1)$ in graphs of bounded $\grad_{(t-1)/2}$ in time $O(m)$, then given a graph $G$ we could construct $\Gt$ and obtain $\Sub{G'}{\cC_k}$ in time $O(m)$, as from \Clm{G_bounded_grad_1} we have that $\Gt$ has bounded $\grad_{(t-1)/2}$. We can then use \Clm{triangles_to_cycles_1} to determine if $G$ contains a triangle. However, this directly contradicts the \TRICONJ\ and hence not such algorithm can exist.

	Similarly, assume that we have an algorithm that can count the number of cycles of size $k=3(t+1)+1$ in graphs of bounded $\grad_{(t-1)/2}$ in time $O(m)$, then given a graph $G$ we could construct $\Gtt$ and obtain $\Sub{G'}{\cC_k}$ in time $O(m)$, as from \Clm{G_bounded_grad_1} we have that $G'$ has bounded $\grad_{(t-1)/2}$. We can then use \Clm{triangles_to_cycles_2} to determine if $G$ contains a triangle. Again, this directly contradicts the \TRICONJ\ and hence not such algorithm can exists.
\end{proof}

\subsection{From Cycle Subgraphs to Homomorphisms}
	
	Now we extend the hardness result from counting subgraphs to counting homomorphisms of cycle graphs. It is given by the following lemma:
	
	\begin{lemma} \label{lem:sub_to_hom_cycles}
	For all $t>1 \in \mathbb{N}$, let $G$ be any input graph with $n$ vertices, $m$ edges and bounded $\grad_{(t-1)/2}(G)$ and let $H = \cC_k$ be the cycle graph on $k$ vertices for $k \in \{3(t+1), 3(t+1)+1, 3(t+1)+2\}$. Assuming  the \TRICONJ, there exists an absolute constant $\gamma > 0$ such that there is no (expected) $o(m^{1+\gamma})$ algorithm for the $\Hom{G}{H}$ problem.
	\end{lemma}
	\begin{proof}
		Fix any $t>1$. We prove each of the cases separately:
		\begin{itemize}
			\item Let $k=3(t+1)$, let $H = \cC_k$ be the cycle with $k$ vertices and $G$ any graph with bounded $\grad_{(t-1)/2}(G)$. Let $H'$ be any graph in the $\Spasm$ of $H$ different than $H$. We have that $\LICL(H') < 3(t+1)-1$, hence by \Thm{upper} we can compute $\Hom{G}{H'}$ in $O(n)$ time. Now, assume that we can compute $\Hom{G}{H}$ in $O(m)$ time, then we could use \Lem{sub} to obtain $\Sub{G}{H}$, but that contradicts \Lem{sub_cycles}, and hence no $O(m)$ algorithm exists.
			
			\item Similarly, let $k=3(t+1)+1$, let $H = \cC_k$ be the cycle with $k$ vertices and $G$ any graph with bounded $\grad_{(t-1)/2}(G)$. Let $H'$ be any graph in the $\Spasm$ of $H$ different than $H$. We have that $\LICL(H') < 3(t+1)$, hence by \Thm{upper} we can compute $\Hom{G}{H'}$ in $O(n)$ time. Now, assume that we can compute $\Hom{G}{H}$ in $O(m)$ time, then we could use \Lem{sub} to obtain $\Sub{G}{H}$, but that contradicts \Lem{sub_cycles}, and hence no $O(m)$ algorithm exists.
			\item Finally, let $k=3(t+1)+2$, let $H = \cC_k$ be the cycle with $k$ vertices and $G$ any graph. Consider the reduced graph $\Gt$, remember that from \Clm{G_bounded_grad_1} we have that $\Gt$ has bounded $\grad_{(t-1)/2}(G)$. $\Gt$ can not contain any cycle of length exactly $k$, as every cycle has a multiple of $(t+1)$ edges, hence $\Sub{\Gt}{H}=0$. Consider the Spasm of $H$, apart from $\cC_k$ itself the only one other pattern in $\Spasm(H)$ with $LICL \geq 3(t+1)$ will be the cycle $\cC_{k-2}$ with a tail, let $H^*$ be such pattern.
			
			For any other pattern $H'\in \Spasm(H)\setminus \{H,H^*\}$ we have that $LICL(H')<3(t+1)$ and hence by \Thm{upper} we can compute $\Hom{G}{H'}$ in $O(n)$ time. Now assume there is a $O(m)$ algorithm that allows us to compute $\Hom{G}{H}$. Then we could use \Lem{sub} to obtain the value of $\Hom{G}{H^*}$ as all the other terms in the equation will be known. However, we just show that counting homomorphisms of the $\cC_{3(t+1)}$ cycle is not possible in linear time, and by \Lem{hom_subgraphs} we will have that there is no algorithm for counting $H^*$ as it is a supergraph of $\cC_{3(t+1)}$.
		\end{itemize}

	\end{proof}
	
	We can now complete the proof of the lower bound:

	\begin{proof}[Proof of Theorem \ref{thm:lower}]
			First, for $t=1$ the theorem is true. As the statement becomes equivalent to show that there is no algorithm for counting cycles of length greater than $6$ in bounded degeneracy graphs(assuming \TRICONJ), this was proved in \cite{Bera2021}.
			
			Hence we just need to prove for $t>1$. Note that suffices to show that there is no $o(m)$ algorithm for computing $\Hom{G}{H}$ for graphs $G$ of bounded $\grad_{(t-1)/2}$ and graphs $H$ with $LICL(H) \in [3(t+1),3(t+2))$.
			
			Fix some $t>1$, and let $H$ bet any graph with $LICL(H) \in [3(t+1),3(t+2))$, note that $H$ must be a supergraph of either $\cC_{3(t+1)},\cC_{3(t+1)+1},$ or $\cC_{3(t+1)+2}$. Now assume that there is an algorithm that computes $\Hom{G}{H}$ in $O(m)$ time for graphs $G$ with bounded $\grad_{(t-1)/2}$, then using \Lem{hom_subgraphs} we have that we can compute $\Hom{G}{H'}$ for $H'\in\{\cC_{3(t+1)},\cC_{3(t+1)+1},\cC_{3(t+1)+2}\}$, but this directly contradicts \Lem{sub_to_hom_cycles}, completing the proof.
	\end{proof}

\section{From Homomorphism to non-induced copies}

\begin{proof}[Proof of Theorem \ref{thm:main-sub}]
	
	We first prove the upper bound: Let $G$ be any input graph with bounded $\grad_{r/2}$ and $H$ a graph with constant size $k$ and $\LICL(\Spasm(H))<3(r+2)$. Consider any graph $H' \in \Spasm(H)$, we have that $\LICL(H')<3(r+2)$, and hence by \Thm{upper} there is an algorithm that computes $\Hom{G}{H'}$ in time $f(\grad_{r/2})O(m)$, for some explicit function $f$. The size of $\Spasm(H)$ only depends on $k$, thus we can compute $\Hom{G}{H'}$ for all the graphs $H' \in \Spasm(H)$ in $f(\grad_{r/2})O(m)$ time. Using \Lem{sub} we have that we can compute $\Sub{G}{H}$ as a linear combination of $\Hom{G}{H'}$ for all the $H'$ in the spasm of $H$, this will take additional constant time giving the upper bound result.
	
	Now we prove the lower bound: Again let $G$ be any input graph with bounded $\grad_{r/2}$ and $H$ a graph with constant size $k$ and $\LICL(\Spasm(H))\geq 3(r+2)$. This means that there exists a graph $H' \in \Spasm(H)$ with $\LICL(H')\geq 3(r+2)$. Assume by contradiction that there is a $O(m)$ algorithm that computes $\Sub{G}{H}$. We can then use \Lem{generalization_4_2} to construct a series of graphs $G_1,\ldots,G_k$ that are also bounded $\grad_{r/2}$. We can then compute $\Sub{G_i}{H}$ for each of the graphs using the $O(m)$ algorithm that we are assuming exists. 
	
	Then, because $\Sub{G_i}{H}$ is a linear combination of $\Hom{G_i}{H'}$ for all the $H'\in \Spasm(H)$ we can apply again \Lem{generalization_4_2} to compute $\Hom{G}{H'}$ for all the $H'\in \Spasm(H)$ in additional constant time. However, recall that there is a $H' \in \Spasm(H)$ for which $\LICL(H') \geq 3(r+2)$. By \Thm{lower} we have that there is no $O(m)$ algorithm to compute $\Hom{G}{H'}$. Hence, we reach a contradiction.
\end{proof}

\bibliographystyle{alpha}
\bibliography{hom_subgraph_counting}

\end{document}